\documentclass[11pt,a4paper,reqno]{amsart}%
\usepackage[latin1]{inputenc}
\usepackage{mathrsfs}
\usepackage{dsfont}
\usepackage{hyperref}
\usepackage{amsmath}
\usepackage{amssymb}
\usepackage{amsthm}
\usepackage{amsfonts}
\usepackage{amstext}
\usepackage{amsopn}
\usepackage{amsxtra}
\usepackage{mathrsfs}
\usepackage{dsfont}
\usepackage{esint}
\usepackage{pst-all}
\usepackage{pstricks}
\usepackage{graphicx}

\theoremstyle{plain}
\newtheorem{theorem}{Theorem}[section]
\newtheorem{lemma}[theorem]{Lemma}
\newtheorem{corollary}[theorem]{Corollary}

\theoremstyle{definition}

\theoremstyle{remark}
\newtheorem{remark}[theorem]{Remark}

\numberwithin{equation}{section}
\newcommand{\dps}{\displaystyle}
\newcommand{\ii}{\infty}
\newcommand\R{{\ensuremath {\mathbb R} }}
\newcommand\C{{\ensuremath {\mathbb C} }}
\newcommand\N{{\ensuremath {\mathbb N} }}
\newcommand\Z{{\ensuremath {\mathbb Z} }}

\newcommand\1{{\ensuremath {\mathds 1} }}
\newcommand\bP{{\ensuremath {\mathds P} }}
\newcommand\bQ{{\ensuremath {\mathds Q} }}
\renewcommand\phi{\varphi}

\newcommand{\gS}{\mathfrak{S}}

\newcommand{\cE}{\mathcal{E}}

\newcommand{\eps}{\epsilon}

\renewcommand{\epsilon}{\varepsilon}
\newcommand\pscal[1]{{\ensuremath{\left\langle #1 \right\rangle}}}
\newcommand{\norm}[1]{ \left| \! \left| #1 \right| \! \right| }
\newcommand{\tr}{{\rm Tr}\,}

\renewcommand{\geq}{\geqslant}
\renewcommand{\leq}{\leqslant}

\renewcommand{\tilde}{\widetilde}

\newcommand{\nn}{\nonumber}

\newcommand{\bDelta}{\mathbf{\Delta}}

\title{Statistical Mechanics of the Uniform Electron Gas}

\author[M. Lewin]{Mathieu Lewin}
\address{CNRS \& CEREMADE, Universit\'e Paris-Dauphine, PSL Research University, F-75016 Paris, France} 
\email{mathieu.lewin@math.cnrs.fr}

\author[E.H. Lieb]{Elliott H. Lieb}
\address{Departments of Mathematics and Physics, Jadwin Hall, Princeton University, Washington Rd., Princeton, NJ 08544, USA}
\email{lieb@princeton.edu}

\author[R. Seiringer]{Robert Seiringer}
\address{IST Austria (Institute of Science and Technology Austria), Am Campus 1, 3400 Klosterneuburg, Austria}
\email{robert.seiringer@ist.ac.at}

\date{November 14, 2017}

\keywords{Uniform electron gas, Density Functional Theory, thermodynamic limit, statistical mechanics, mean-field limit, optimal transport}

\begin{document}

\maketitle

\begin{abstract}
In this paper we define and study the classical Uniform Electron Gas (UEG), a system of infinitely many electrons whose density is constant everywhere in space. The UEG is defined differently from Jellium, which has a positive constant background but no constraint on the density. We prove that the UEG arises in Density Functional Theory in the limit of a slowly varying density, minimizing the indirect Coulomb energy. We also construct the quantum UEG and compare it to the classical UEG at low density. 

\bigskip

\noindent \sl \copyright~2017 by the authors. Final version to appear in \emph{J.~\'Ec.~polytech.~Math.} This paper may be reproduced, in its entirety, for non-commercial purposes.
\end{abstract}


\tableofcontents

\section{Introduction}
The \emph{Uniform Electron Gas} (UEG) is a cornerstone of Density Functional Theory (DFT)~\cite[Sec.~1.5]{PerKur-03}. This system appears naturally in the regime of slowly varying densities and it is used in the \emph{Local Density Approximation} of DFT~\cite{HohKoh-64,KohSha-65,PerKur-03}. In addition, it is a reference system for most of the empirical functionals used today in DFT, which are often exact for constant densities~\cite{Perdew-91,PerWan-92,Becke-93,PerBurErn-96,SunPerRuz-15,Perdew_etal-16}. 

In this paper, we define the UEG by the property that it minimizes the many-particle Coulomb energy, and satisfies the additional constraint that its electronic density is exactly constant over the whole space. In the literature the UEG is often identified with \emph{Jellium} (or \emph{one-component plasma}) which is defined differently. Jellium has an external constant background, introduced to compensate the repulsion between the particles, and no particular constraint on the density. The Jellium ground state minimizes an energy which incorporates the external potential of the background, in addition to the many-particle Coulomb energy. This ground state usually does \emph{not} have a constant density since it is believed to form a Wigner crystal. But one can always average over the position of this crystal (a state sometimes called the \emph{floating crystal}~\cite{BisLuh-82,MikZie-02,DruRadTraTowNee-04}) and get a constant density, hence the confusion. 

In a recent paper~\cite{LewLie-15}, the first two authors of this article have questioned the identification of the UEG with Jellium in the Coulomb case. They have shown on an example that the averaging does not commute with the thermodynamic limit: the indirect energy per unit volume of a floating crystal can be much higher than its Jellium energy. Hence it is not clear if the floating crystal is a minimizer at constant density. These pecularities are specific to the Coulomb case and they have been discussed before in several works~\cite{NavJamFei-80,ChoFavGru-80,BorBorShaZuc-88,BorBorSha-89,BorBorStr-14}.

It is not the purpose of this paper to answer the important question of whether Jellium and the UEG are the same or not. Our goal here is, rather, to properly define the UEG using tools from statistical mechanics and to provide some of its properties. Although there are many rigorous results on the statistical mechanics of Jellium-like systems (see, e.g.~\cite{Kunz-74,LieNar-75,BraLie-75,AizMar-80,FroPar-78,ChoFavGru-80,GruLugMar-80,GruMar-81,Imbrie-82,BryMar-99,SanSer-15,PetSer-15,RouSer-16,LebSer-17}), our work seems to be the first mathematical discussion of the UEG.

In this paper, we concentrate much of our attention on the \emph{classical} UEG, which is often called \emph{strongly} or \emph{strictly correlated} since it appears in a regime where the interaction dominates the kinetic energy, that is, at low density. The classical UEG has been the object of many recent numerical works, based on methods from optimal transportation~\cite{Seidl-99,SeiPerLev-99,SeiGorSav-07,GorSei-10,SeiMarGerNenGieGor-17}. In addition to providing interesting properties of DFT at low density, the classical UEG has been used to get numerical bounds on the best constant in the Lieb-Oxford inequality~\cite{Lieb-79,LieOxf-80,LieSei-09}. This universal lower bound on the Coulomb energy for finite and infinite systems is also used in the construction of some DFT functionals~\cite{Perdew-91,PerBurErn-96,SunPerRuz-15,SunRuzPer-15,Perdew_etal-16}.

We now give a short description of our results. The \emph{indirect Coulomb energy} of a given density $\rho(x)\geq0$ with $\int_{\R^3}\rho(x)\,dx=N$ is defined by 
\begin{multline}
E(\rho):=\inf_{\substack{\bP \text{\ $N$-particle}\\ 
\text{probability on $\R^{3N}$}\\\rho_\bP=\rho}}\bigg\{\int_{\R^{3N}}\sum_{1\leq j<k\leq N}\frac{1}{|x_j-x_k|}d\bP(x_1,...,x_N)\bigg\}\\-\frac12\int_{\R^3}\int_{\R^3}\frac{\rho(x)\rho(y)}{|x-y|}dx\,dy,
\label{eq:def_E_intro}
\end{multline}
where $\rho_\bP$ is by definition the sum of the one-particle marginals of $\bP$. Note that the infimum takes the form of a multi-marginal optimal transportation problem~\cite{CotFriKlu-13,CotFriPas-15,MarGerNen-17,SeiMarGerNenGieGor-17}. 
The (classical) UEG ground state energy is obtained by imposing the constraint that $\rho$ is constant over a set $\Omega_N\subset\R^3$ with $|\Omega_N|=N$ and taking the thermodynamic limit
\begin{equation}
 e_{\rm UEG}=\lim_{\Omega_N\nearrow\,\R^3}\frac{E(\1_{\Omega_N})}{N}.
 \label{eq:def_limit_thermo_intro}
\end{equation}
After taking the limit we obtain a density which is constant in the whole space, here equal to $\rho_\ii\equiv1$. By scaling, the energy at constant $\rho_\ii$ is given by $e_{\rm UEG}\,\rho_\ii^{4/3}$. Using well-known tools from statistical mechanics~\cite{Ruelle,LieSei-09}, we show below in Section~\ref{sec:limit_thermo} that the limit~\eqref{eq:def_limit_thermo_intro} exists and is independent of the chosen sequence $\Omega_N$, provided the latter has a sufficiently regular boundary. The reader can just think of $\Omega_N$ being a sequence of balls or cubes, or any scaled convex set. Our argument relies on the subadditivity of the classical indirect energy~\eqref{eq:def_E_intro}, that is,
\begin{equation}
 E(\rho_1+\rho_2)\leq E(\rho_1)+E(\rho_2)
 \label{eq:subadditivity_intro}
\end{equation}
for all densities $\rho_1,\rho_2\geq0$ (see Lemma~\ref{lem:subadditivity}).

After having properly defined the UEG energy~\eqref{eq:def_limit_thermo_intro}, we prove in Section~\ref{sec:limit_slowly_varying_rho} that it arises in the limit of slowly varying densities. Namely, we show in Theorem~\ref{thm:rho_scaled} below that 
\begin{equation}
\lim_{N\to\ii} \frac{E\big(\rho(\cdot/N^{1/3})\big)}{N}= \lim_{N\to\ii} \frac{E\big(N\rho\big)}{N^{4/3}}=e_{\rm UEG}\int_{\R^3}\rho(x)^{4/3}\,dx,
\label{eq:limit_slowly_varying_rho_intro}
\end{equation}
for any fixed density $\rho(x)$ with $\int_{\R^3}\rho(x)\,dx=1$.
This limit has been the object of recent numerical works~\cite{RasSeiGor-11,SeiVucGor-16}. That $E\big(N\rho\big)=O(N^{4/3})$ follows immediately from the Lieb-Oxford inequality, which we will recall below, as was already remarked in~\cite[Rem.~1.5]{CotFriPas-15}. Based on the limit~\eqref{eq:limit_slowly_varying_rho_intro}, one can use any density $\rho$ in order to compute an approximation of $e_{\rm UEG}$. 
In~\cite{SeiVucGor-16} it was observed that the limit~\eqref{eq:limit_slowly_varying_rho_intro} seems to be attained faster for smoother densities than it is for a characteristic function $\1_{\Omega_N}$ appearing in~\eqref{eq:def_limit_thermo_intro}. 

The interpretation of~\eqref{eq:limit_slowly_varying_rho_intro} is the following. If we think of splitting the space $\R^3$ using a tiling made of cubes of side length $1\ll \ell \ll N^{1/3}$, we see that $\rho(x/N^{1/3})$ is essentially constant in each of these large cubes. The local energy can therefore be replaced by $e_{\rm UEG}(\rho_k)^{4/3}$ where $\rho_k$ is the average value of $\rho$ in the $k$th cube. The energy $E$ is however not local and there are interactions between the different cubes. Proving~\eqref{eq:limit_slowly_varying_rho_intro} demands to show that these interaction energies do not appear at the leading order.

Our proof of~\eqref{eq:limit_slowly_varying_rho_intro} requires us to extend the definition~\eqref{eq:def_limit_thermo_intro} of the UEG energy to grand canonical states, that is, to let the particle number $N$ fluctuate. The reason for this is simple. In spite of the fact that the total particle number $N$ is fixed, the number of particles in a set $A\subset \R^3$ (for instance a cube of side length $\ell$ as before) is not known exactly. This number can fluctuate around its average value $\int_A\rho$, and these fluctuations influence the interactions between the cubes. In Section~\ref{sec:grand-canonical}, we therefore give a proper definition of the grand-canonical UEG and prove that its thermodynamic limit is the same as in~\eqref{eq:def_limit_thermo_intro}.

Like for Jellium~\cite{GruLugMar-78,GruLugMar-80,MarYal-80,GruLebMar-81}, it is to be expected that the long range nature of the Coulomb potential will reduce the fluctuations, due to screening. Following previous works for Coulomb systems in~\cite{HaiLewSol_1-09,HaiLewSol_2-09,BlaLew-12}, we use the Graf-Schenker inequality~\cite{GraSch-95} to exhibit this effect and conclude the proof of~\eqref{eq:limit_slowly_varying_rho_intro}.

In Section~\ref{sec:quantum} we finally look at the quantum case. Proving the existence of the thermodynamic limit similar to~\eqref{eq:def_limit_thermo_intro} in the quantum case is much more difficult since the quantum energy does not satisfy the subadditivity property~\eqref{eq:subadditivity_intro}. Our proof follows the method introduced in~\cite{HaiLewSol_1-09,HaiLewSol_2-09}, which is also based on the Graf-Schenker inequality. For completeness, we also prove that the classical UEG is obtained in the low-density limit $\rho\to0$ (or equivalently the semi-classical limit $\hbar\to0$). This seems open so far for finite systems, except when $N=2,3$~\cite{CotFriKlu-13,BinPas-17}. At high density, we use a result by Graf-Solovej~\cite{GraSol-94} to deduce that the quantum energy behaves as 
$$c_{\rm TF}\,\rho^{5/3}-c_{\rm D}\rho^{4/3}+o(\rho^{4/3})_{\rho\to\ii}$$
where $c_{\rm TF}$ and $c_{\rm D}$ are, respectively, the Thomas-Fermi and Dirac constants.

Many of our results are valid in a more general setting. For completeness we properly define the classical UEG for general Riesz-type potentials
$$V(x)=\frac{1}{|x|^s}$$
in $\R^d$, with $0<s<d$, although we are more interested in the physical case $s=1$ in dimension $d=3$. Several of our results (the limit~\eqref{eq:limit_slowly_varying_rho_intro} as well as the quantum problem) actually only hold in the case $s=1$ and $d=3$. Extending our findings to other values of $s$ and other dimensions is an interesting question which could shed light on the specific properties of the Coulomb potential, in particular with regards to screening effects.

\subsection*{Acknowledgement} We thank Paola Gori Giorgi who has first drawn our attention to this problem, as well as Codina Cotar, Simone Di Marino and Mircea Petrache for useful discussions. We also thank the \emph{Institut Henri Poincar\'e} in Paris for its hospitality. This project has received funding from the European Research Council (ERC) under the European Union's Horizon 2020 research and innovation programme (grant agreement 694227 for R.S. and MDFT 725528 for M.L.). Financial support by the Austrian Science Fund (FWF), project No P 27533-N27 (R.S.) and by the US National Science Foundation, grant No PHY12-1265118 (E.H.L.) are gratefully acknowledged.

\section{Definition of the Uniform Electron Gas}

\subsection{The indirect energy and the Lieb-Oxford inequality}

Everywhere in the paper, we deal with the Riesz interaction potential
$$V(x)=\frac{1}{|x|^s}$$
in $\R^d$, except when explicitly mentioned. Several of our results will only hold for $s=1$ and $d=3$, which is the physical Coulomb case. We always assume that 
$$0<s<d$$
such that $V$ is locally integrable in $\R^d$. The $d$-dimensional Coulomb case corresponds to $s=d-2$ for $d\geq2$. 

In dimension $d=2$ the case $s=0$ is formally obtained by expanding $V$ as $s\to0^+$, leading to the potential $V(x)=-\log|x|$. In dimension $d=1$ one can go down to $-1\leq s<0$ with $V(x)=-|x|^{|s|}$ where $s=-1$ is the Coulomb case. In these situations the potential $V$ diverges to $-\ii$ at large distances. For simplicity we will not consider these cases in detail and will only make some short comments without proofs.

Let $\rho\geq0$ be a non-negative function on $\R^d$, with $\int_{\R^d}\rho=N$ (an integer) and $\rho\in L^{\frac{2d}{2d-s}}(\R^d)$. The indirect energy of $\rho$ is by definition the lowest classical exchange-correlation energy that can be reached using $N$-particle probability densities having this density $\rho$. In other words, 
\begin{multline}
E(\rho):=\inf_{\substack{\bP \text{\ $N$-particle}\\ 
\text{probability on $\R^{dN}$}\\\rho_\bP=\rho}}\bigg\{\int_{\R^{dN}}\sum_{1\leq j<k\leq N}\frac{1}{|x_j-x_k|^s}d\bP(x_1,...,x_N)\bigg\}\\-\frac12\int_{\R^d}\int_{\R^d}\frac{\rho(x)\rho(y)}{|x-y|^s}dx\,dy.
\label{eq:def_E}
\end{multline}
The density of the $N$-particle probability $\bP$ is defined by 
\begin{multline*}
\rho_\bP(y)=\int_{\R^{d(N-1)}}d\bP(y,x_2,...,x_N)+\int_{\R^{d(N-1)}}d\bP(x_1,y,...,x_N)+\cdots\\\cdots +\int_{\R^{d(N-1)}}d\bP(x_1,x_2,...,y). 
\end{multline*}
The condition that $\rho\in L^{\frac{2d}{2d-s}}(\R^d)$ guarantees that 
$$\int_{\R^d}\int_{\R^d}\frac{\rho(x)\rho(y)}{|x-y|^s}dx\,dy<\ii$$
by the Hardy-Littlewood-Sobolev inequality~\cite{LieLos-01}. Then $E(\rho)$ is well defined and finite. We will soon assume that $\rho\in L^{1+\frac{s}{d}}(\R^3)$, which is stronger by Hölder's inequality.

Since the many-particle interaction is symmetric with respect to permutations of the variables $x_j$, the corresponding energy is unchanged when $\bP$ is replaced by the symmetrized probability 
$$\tilde\bP(x_1,...,x_N)=\frac1{N!}\sum_{\sigma\in\gS_N}\bP(x_{\sigma(1)},...,x_{\sigma(N)}).$$ 
Since $\rho_{\tilde\bP}=\rho_\bP$ it is clear that we can restrict ourselves to symmetric probabilities $\bP$, without changing the value of the infimum in~\eqref{eq:def_E}. For a symmetric probability we simply have
$$\rho_\bP(y)=N\int_{\R^{d(N-1)}}d\bP(y,x_2,...,x_N).$$
It will simplify some arguments to be able to consider non-symmetric probabilities $\bP$.

In the following, we use the notation
$$C(\bP):=\int_{\R^{dN}}\sum_{1\leq j<k\leq N}\frac{1}{|x_j-x_k|^s}d\bP(x_1,...,x_N)=\pscal{\sum_{1\leq j<k\leq N}\frac{1}{|x_j-x_k|^s}}_\bP$$
for the many-particle energy and
$$D(f,g):=\frac12\int_{\R^d}\int_{\R^d}\frac{\overline{f(x)}g(y)}{|x-y|^s}dx\,dy$$
for the direct term. We recall that 
$$D(f,f)=c_{d,s}\int_{\R^d}\frac{|\widehat{f}(k)|^2}{|k|^{d-s}}\,dk\geq0 \qquad\text{with}\qquad c_{d,s}=\frac{2^{d-1-s}\pi^{\frac{d}2}\,\Gamma\!\left(\frac{d-s}{2}\right)}{\Gamma\!\left(\frac{s}{2}\right)}$$
defines an inner product.

Taking $\bP=(\rho/N)^{\otimes N}$ as trial state, we find the simple upper bound 
$$E(\rho)\leq -\frac{1}{N}D(\rho,\rho)<0.$$
On the other hand, the Lieb-Oxford inequality~\cite{Lieb-79,LieOxf-80,LieSei-09} gives a useful lower bound on $E(\rho)$, under the additional assumption that $\rho\in L^{1+\frac{s}{d}}(\R^d)$.

\begin{theorem}[Lieb-Oxford inequality~\cite{Lieb-79,LieOxf-80,LieSei-09,Bach-92,GraSol-94,LieSolYng-95,LunNamPor-16}]
Assume that $0<s<d$ in dimension $d\geq1$. Then there exists a universal constant $C_{\rm LO}(s,d)>0$ such that 
\begin{equation}
 E(\rho)\geq -C_{\rm LO}(s,d)\int_{\R^d}\rho(x)^{1+\frac{s}{d}}\,dx,
 \label{eq:LO}
\end{equation}
for every $\rho\in L^{1+\frac{s}{d}}(\R^d)$.
\end{theorem}

From now on we always call $C_{\rm LO}(s,d)$ the \emph{smallest} constant for which the inequality~\eqref{eq:LO} is valid. 

Although only the case $s=1$ and $d=3$ was considered in the original papers~\cite{Lieb-79,LieOxf-80}, the proof for $s=1$ and $d=2$ given in~\cite{Bach-92,GraSol-94,LieSolYng-95} extends to any $0<s<d$ in any dimension, see~\cite[Lemma~16]{LunNamPor-16}. The proof involves the Hardy-Littlewood estimate for the maximal function $M_\rho$,
$$\norm{M_\rho}_{L^{1+s/d}(\R^d)}\leq C_{\rm HL}(s,d)\norm{\rho}_{L^{1+s/d}(\R^d)}$$
and, consequently,  the best known estimate on $C_{\rm LO}(s,d)$ involves the unknown constant $C_{\rm HL}(s,d)$.

In the 3D Coulomb case, $d=3$ and $s=1$, the best estimate known so far is
$$C_{\rm LO}(1,3)\leq 1.64.$$
The constant was equal to $1.68$ in~\cite{LieOxf-80} and later improved to $1.64$ in~\cite{ChaHan-99}. It remains an important challenge to find the optimal constant in~\eqref{eq:LO}. Several of the most prominent functionals used in Density Functional Theory are actually based on the Lieb-Oxford bound~\cite{Perdew-91,PerBurErn-96,SunPerRuz-15,SunRuzPer-15,Perdew_etal-16}. 

The best rigorous lower bound on $C_{\rm LO}(1,3)$ was proved already in~\cite{LieOxf-80} for $N=\int_{\R^3}\rho=2$:
$$C_{\rm LO}(1,3)\geq 1.23,$$
whereas the latest numerical simulations in~\cite{SeiVucGor-16} give the estimate
$$C_{\rm LO}(1,3) \gtrsim1.41.$$
From the definition of the UEG energy given later in~\eqref{eq:thermodynamic_limit} it will be clear that
\begin{equation}
 C_{\rm LO}(s,d)\geq -e_{\rm UEG}.
 \label{eq:compare_LO_UEG}
\end{equation}
It has indeed been conjectured in~\cite{OdaCap-07,RasPitCapPro-09} that the best Lieb-Oxford constant is attained for the \emph{Uniform Electron Gas} (UEG), that is, there is equality in~\eqref{eq:compare_LO_UEG}. 

Next we turn to some remarks about the Lieb-Oxford inequality for $s\leq0$.

\begin{remark}[2D Coulomb case]
In~\cite[Prop.~3.8]{LewNamSerSol-15}, the following Lieb-Oxford-type inequality was shown in dimension $d=2$:
\begin{multline}
 -\int_{\R^{2N}}\sum_{1\leq j<k\leq N}\log|x_j-x_k|\,d\bP(x_1,...,x_N)\\+\frac12\int_{\R^2}\int_{\R^2}\rho_\bP(x)\rho_\bP(y)\log|x-y|\,dx\,dy\\
 \geq-\frac{1}{4}N\log N-\frac{c}{\eps}N-\frac{c}{\eps}N^{-\eps}\int_{\R^2}\rho(x)^{1+\eps}\,dx,
 \label{eq:LO_2D}
\end{multline}
for $\int_{\R^2}\rho=N$, any $\eps>0$ and some constant $c>0$. This inequality can be used to deal with the case $s=0$ and $d=2$. For shortness, we do not elaborate more on the 2D Coulomb case.
\end{remark}

\begin{remark}[1D case]
In dimension $d=1$, we have for $-1\leq s<0$ the Lieb-Oxford inequality with $C_{\rm LO}(s,d)=0$,
\begin{multline}
-\int_{\R^{N}}\sum_{1\leq j<k\leq N}|x_j-x_k|^{|s|}\,d\bP(x_1,...,x_N)\\+\frac12\int_{\R}\int_{\R}\rho_\bP(x)\rho_\bP(y)|x-y|^{|s|}\,dx\,dy\geq0.
\label{eq:LO_1D}
\end{multline}
Indeed, 
$$-\int_{\R}\int_{\R}|x-y|^{|s|}\,d\nu(x)\,d\nu(y)=c\int_{\R}\frac{|\widehat{\nu}(k)|^2}{|k|^{1+|s|}}\,dk\geq0$$
for every bounded measure $\nu$ which decays sufficiently fast at infinity and satisfies $\nu(\R)=0$. After taking $\nu=\sum_{j=1}^N\delta_{x_j}-\rho_\bP$ we find the pointwise bound
\begin{multline*}
-\sum_{1\leq j<k\leq N}|x_j-x_k|^{|s|}+\sum_{j=1}^N\int_\R |x_j-y|^{|s|}\rho_\bP(y)\,dy\\
-\frac12\int_{\R}\int_{\R}\rho_\bP(x)\rho_\bP(y)|x-y|^{|s|}\,dx\,dy \geq0
\end{multline*}
in $\R^N$. Integrating against $\bP(x_1,...,x_N)$ gives~\eqref{eq:LO_1D}.
\end{remark}

\begin{remark}[1D multi-marginal optimal transport]
For every fixed $\rho\in L^1(\R)\cap L^{1+s}(\R)$ with $\int_\R\rho=N$ and $0<s<1$, the minimization problem~\eqref{eq:def_E} has been proved in~\cite{ColPasMar-15} to have a unique symmetric minimizer $\bP$ of the form
\begin{equation}
d\bP(x_1,...x_N)=\frac{1}{N!}\sum_{\sigma\in\gS_N}\int_\R \delta_y(x_{\sigma(1)})\delta_{Ty}(x_{\sigma(2)})\cdots \delta_{T^{N-1}y}(x_{\sigma(N)})\,\rho(y)\,dy.
\label{eq:Monge_1D}
\end{equation}
Here $T:\R\to\R$ is the unique increasing function such that $T\#(\rho\1_{(a_{k-1},a_k)})=\rho\1_{(a_{k},a_{k+1})}$, where the numbers $a_k$ are defined by $a_0=-\ii$ and $\int_{a_k}^{a_{k+1}}\rho=1$. 
For consequences with regards to the Lieb-Oxford inequality and the uniform electron gas, we refer to~\cite{DiMarino-17}.
\end{remark}

\subsection{Subadditivity}
Before turning to the special case of constant density, we state and prove an important property of the indirect energy $E$, which will be used throughout the paper. 

\begin{lemma}[Subadditivity of the indirect energy]\label{lem:subadditivity}
Let $\rho_1,\rho_2\in L^1(\R^d)\cap L^{1+s/d}(\R^d)$ be two positive densities with $\int_{\R^d}\rho_1=N_1$ and $\int_{\R^d}\rho_2=N_2$ (two integers). Then
\begin{equation}
E(\rho_1+\rho_2)\leq E(\rho_1)+E(\rho_2).
\label{eq:subadditivity}
\end{equation}
\end{lemma}

\begin{proof}
Let $\bP_1$ and $\bP_2$ be two $N_1$-- and $N_2$--particle probabilities of densities $\rho_1$ and $\rho_2$. We use as trial state the uncorrelated probability $\bP=\bP_1\otimes\bP_2$ defined by
\begin{align*}
\bP(x_1,...,x_{N_1+N_2})&=\bP_1\otimes\bP_2(x_1,...,x_{N_1+N_2})\\
&=\bP_1(x_1,...,x_{N_1})\bP_2(x_{N_1+1},...,x_{N_1+N_2}). 
\end{align*}
Even if $\bP_1$ and $\bP_2$ are symmetric $\bP$ is not necessarily symmetric, but it can be symmetrized without changing anything if the reader feels more comfortable with symmetric states. The density of this trial state is computed to be $\rho_\bP=\rho_{1}+\rho_{2}$ and the many-particle energy is
\begin{align*}
C(\bP)=&\pscal{\sum_{1\leq j<k\leq N_1}\frac{1}{|x_j-x_k|^s}}_{\bP_1\otimes\bP_2}+\pscal{\sum_{N_1+1\leq j<k\leq N_2}\frac{1}{|x_j-x_k|^s}}_{\bP_1\otimes\bP_2}\\
&+\pscal{\sum_{j=1}^{N_1}\sum_{k=N_1+1}^{N_1+N_2}\frac{1}{|x_j-x_k|^s}}_{\bP_1\otimes\bP_2}\\
=&C(\bP_1)+C(\bP_2)+2D(\rho_1,\rho_2),
\end{align*}
and hence
$$E(\rho_1+\rho_2)\leq C(\bP)-D(\rho_1+\rho_2,\rho_1+\rho_2)=C(\bP_1)-D(\rho_1,\rho_1)+C(\bP_2)-D(\rho_2,\rho_2).$$
Optimizing over $\bP_1$ and $\bP_2$ gives the result. 
\end{proof}

\subsection{The Classical Uniform Electron Gas}\label{sec:limit_thermo}

Next we define the (classical) \emph{Uniform Electron Gas} (UEG) corresponding to taking $\rho=\rho_0\1_{\Omega}$ (the characteristic function of a domain $\Omega$) and then the limit when $\Omega$ covers the whole of $\R^3$. For this it is useful to discuss regularity of sets. Following Fisher~\cite{Fisher-64}, we say that a set $\Omega$ has an $\eta$--regular boundary when
\begin{equation}
 \forall 0\leq t\leq t_0,\qquad \Big|\big\{{\rm d}(x,\partial\Omega)\leq |\Omega|^{1/d}t\big\}\Big|\leq |\Omega|\,\eta(t). 
 \label{eq:Fisher}
\end{equation}
Here $t_0>0$ and $\eta$ is a continuous function $\eta:[0,t_0)\to\R^+$ with $\eta(0)=0$. Note that the definition is invariant under scaling. If $\Omega$ has an $\eta$--regular boundary, then the dilated set $\lambda\Omega$ does as well for all $\lambda>0$. The concept of $\eta$--regularity allows to make sure that the area of the boundary of $\Omega$ is negligible compared to $|\Omega|$. Note that any open convex domain (e.g. a ball or a cube) has an $\eta$--regular boundary with $\eta(t)=Ct$, see~\cite[Lemma~1]{HaiLewSol_1-09}.

\begin{theorem}[Uniform Electron Gas]\label{thm:UEG_thermo_limit}
Let $\rho_0>0$ and $\{\Omega_N\}\subset\R^d$ be a sequence of bounded connected domains such that 
\begin{itemize}
 \item $\rho_0|\Omega_N|$ is an integer for all $N$;
 \item $|\Omega_N|\to\ii$;
 \item $\Omega_N$ has an $\eta$--regular boundary for all $N$, for some $\eta$ which is independent of $N$.
\end{itemize}
Then the following thermodynamic limit exists
\begin{equation}
 \lim_{N\to\ii}\frac{E(\rho_0\1_{\Omega_N})}{|\Omega_N|}=\rho_0^{1+\frac{s}{d}}\,e_{\rm UEG} 
 \label{eq:thermodynamic_limit}
\end{equation}
where $e_{\rm UEG}$ is independent of the sequence $\{\Omega_N\}$ and of $\rho_0$.
\end{theorem}

The limit~\eqref{eq:thermodynamic_limit} is our definition of the Uniform Electron Gas energy $e_{\rm UEG}$. Since by the Lieb-Oxford inequality~\eqref{eq:LO} we have
\begin{equation}
 E(\rho_0\1_{\Omega})\geq -C_{\rm LO}(s,d)(\rho_0)^{1+\frac{s}{d}}|\Omega| 
 \label{eq:LO_UEG}
\end{equation}
for any domain $\Omega$, it is clear from the definition~\eqref{eq:thermodynamic_limit} that
$$e_{\rm UEG}\geq -C_{\rm LO}(s,d),$$
as we have mentioned before.

Our proof of Theorem~\ref{thm:UEG_thermo_limit} follows classical methods in statistical mechanics~\cite{Ruelle,LieSei-09}, based on the subadditivity property~\eqref{eq:subadditivity}. 

\begin{proof}
Everywhere we use the shorthand notation $E(\Omega)=E(\1_\Omega)$. 

\medskip

\noindent\textbf{Step 1. Scaling out $\rho_0$.}
By scaling we have $E(\rho_0\1_{\Omega_N})=\rho_0^{s/d} E(\1_{\Omega'_N})$ where $\Omega_N'=\rho_0^{1/d}\Omega_N$, which is also regular in the sense of~\eqref{eq:Fisher}. Therefore, it suffices to prove the theorem for $\rho_0=1$. 

\medskip

\noindent\textbf{Step 2. Limit for cubes.} Let $C$ be the unit cube and $C_N=2^NC$. Since $C_{N}$ is the union of $2^d$ disjoint copies of the cube $C_{N-1}$, we have by subadditivity
$$E(C_N)\leq 2^dE(C_{N-1})$$
and therefore
$$\frac{E(C_{N})}{|C_{N}|}\leq \frac{E(C_{N-1})}{|C_{N-1}|}.$$
The sequence $E(C_{N})/|C_{N}|$ is decreasing and bounded from below due to the Lieb-Oxford inequality~\eqref{eq:LO_UEG}. Hence it converges to a limit $e_{\rm UEG}$. 

\medskip

The proof that the limit is the same for a general sequence $\Omega_N$ satisfying the assumptions of the theorem is very classical. The idea is to approximate $\Omega_N$ from inside by the union of smaller cubes of side length $\ell\ll|\Omega_N|^{1/d}$, which gives an upper bound by subadditivity. For the lower bound, one uses a large cube $C'_N$ containing $\Omega_N$, of comparable volume, with the space $C'_N\setminus \Omega_N$ filled with small cubes, see Figure~\ref{fig:proof_thermo_limit}. We start with the upper bound.

\medskip

\noindent\textbf{Step 3. Upper bound for any domain $\Omega_N$.} 
For any fixed $n$, we look at the tiling of $\R^d=\cup_j D_j$ with cubes $D_j$ of volume $2^{dn}$ which are all translates of the cube $C_{n}$ considered in the previous step. For simplicity we let $\ell=2^n$ be the side length of this cube. The set $\tilde\Omega_N=\cup_{D_j\subset\Omega_N} D_j$ is an inner approximation of $\Omega_N$ which satisfies
\begin{align*}
|\tilde\Omega_N|&=|\Omega_N|-\sum_{D_j\cap\Omega_N^c\neq\emptyset}|D_j\cap\Omega_N|\\
&\geq |\Omega_N|- \Big|\big\{{\rm d}(x,\partial\Omega_N)\leq \ell\sqrt{d} \big\}\Big|\\
&\geq |\Omega_N|\left(1-\eta(\ell\sqrt{d} |\Omega_N|^{-1/d})\right)
\end{align*}
since the cubes intersecting the boundary only contain points which are at a distance $\leq \ell\sqrt{d}$ to $\partial\Omega_N$. 
By subadditivity we have
$$E(\tilde \Omega_N)\leq \frac{|\tilde\Omega_N|}{|C_{n}|}E(C_{n})$$
where ${|\tilde\Omega_N|}/{|C_{n}|}$ is the number of cubes in $\tilde\Omega_N$. Since $E(\Omega_N\setminus\tilde \Omega_N)<0$, we have
$$\frac{E(\Omega_N)}{|\Omega_N|}\leq \frac{E(\tilde \Omega_N)+E(\Omega_N\setminus\tilde \Omega_N)}{|\Omega_N|}\leq \frac{E(C_{n})}{|C_{n}|}\left(1-\eta(\ell\sqrt{d} |\Omega_N|^{-1/d})\right).$$
Passing to the limit first $N\to\ii$, using $\eta(t)\to0$ when $t\to0$, and then $n\to\ii$ (or taking the joint limit with $\ell\ll |\Omega_N|^{1/d}$) gives the upper bound
$$\limsup_{N\to\ii}\frac{E(\Omega_N)}{|\Omega_N|}\leq e_{\rm UEG}.$$

\medskip

\noindent\textbf{Step 4. Lower bound for any domain $\Omega_N$.} Since we have assumed that our sets are connected, by~\cite[Lemma~1]{Fisher-64} we know that the diameter of $\Omega_N$ is of the order $|\Omega_N|^{1/d}$. Hence $\Omega_N$ is included in a large cube $C'_N$ of side length proportional to $|\Omega_N|^{1/d}$. Increasing this cube if necessary and after a space translation, we can assume that $C'_N=2^kC$ which we have used before, and which is the union of $2^{d(k-n)}$ small cubes $D_j$. Let then $A_N$ be the union of all the small cubes $D_j$ which are contained in $C'_N\setminus \Omega_N$. We write
$$C'_N=A_N\cup \Omega_N\cup R_N$$
where $R_N$ is the missing space (the union of the sets $D_j\cap(C'_N\setminus\Omega_N)$ for all the cubes that intersect the boundary $\partial\Omega_N$, see Figure~\ref{fig:proof_thermo_limit}). Since all these cubes only contain points which are at a distance $\leq \sqrt{d} \ell$ from the boundary of $\Omega_N$, we have as before
$$|R_N|\leq |\Omega_N|\eta(\ell\sqrt{d} |\Omega_N|^{-1/d}).$$
The subadditivity and monotonicity of the energy per unit volume for cubes give that
$$e_{\rm UEG}|C'_N|\leq E(C'_N)\leq E(\Omega_N)+\frac{|A_N|}{|C_{n}|}E(C_{n})+\underbrace{E(R_N)}_{\leq0}$$
and thus
\begin{align*}
\frac{E(\Omega_N)}{|\Omega_N|}\geq e_{\rm UEG}+\frac{|A_N|}{|\Omega_N|}\left(e_{\rm UEG}-\frac{E(C_{n})}{|C_{n}|}\right)+e_{\rm UEG}\,\eta(\ell\sqrt{d} |\Omega_N|^{-1/d}).
\end{align*}
Using that $|A_N|\leq |C'_N|=O(|\Omega_N|)$ and passing to the limit $N\to\ii$ then $n\to\ii$ gives
$$\liminf_{N\to\ii}\frac{E(\Omega_N)}{|\Omega_N|}\geq e_{\rm UEG}$$
as we wanted.
\end{proof}

\begin{figure}[h]
\includegraphics{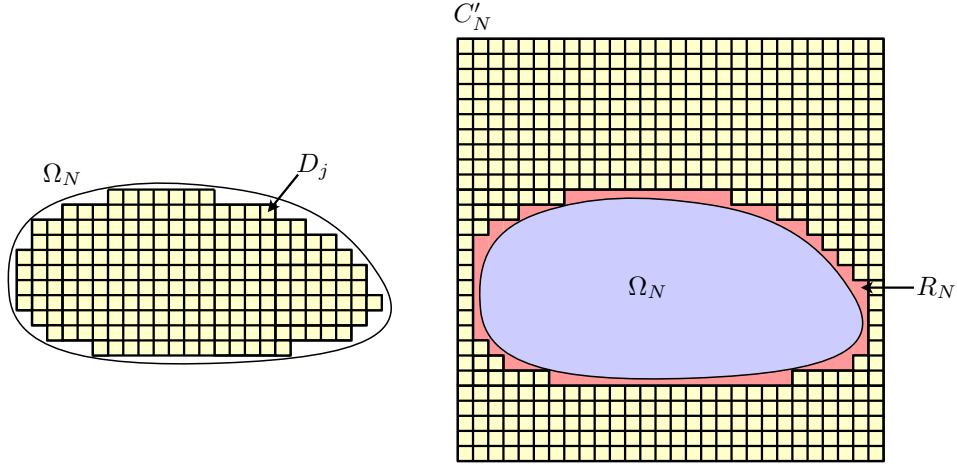}

\caption{Method to prove that the limit for a general domain $\Omega_N$ is the same as for the cubes $C_n=2^nC$, using the subadditivity property~\eqref{eq:subadditivity}. For the upper bound (left), one uses the union $\tilde\Omega_N$ of all the cubes $D_j$ which are inside $\Omega_N$. For the lower bound (right), one uses a big cube $C'_N$ containing $\Omega_N$ with the space between the two filled with smaller cubes.\label{fig:proof_thermo_limit}}
\end{figure}

In the physical case $s=1$ and $d=3$, we have the following lower bound.

\begin{theorem}[Lower bound on $e_{\rm UEG}$~\cite{LieNar-75}]
Assume that $s=1$ and $d=3$. Then we have
\begin{equation}
e_{\rm UEG}\geq -\frac35 \left(\frac{9\pi}2\right)^{1/3}\simeq -1.4508.
\label{eq:lower_bd_UEG}
\end{equation}
\end{theorem}

\begin{proof}
Let $\bP$ be any probability distribution such that $\rho_\bP=\1_\Omega$. Then we can write
\begin{multline*}
\pscal{\sum_{1\leq j<k\leq N}\frac{1}{|x_j-x_k|}}_\bP-D(\1_\Omega,\1_\Omega)\\=\pscal{\sum_{1\leq j<k\leq N}\frac{1}{|x_j-x_k|}-\sum_{j=1}^N\1_\Omega\ast\frac{1}{|\cdot|}(x_j)}_\bP+D(\1_\Omega,\1_\Omega).
\end{multline*}
If we minimize the right side over all probability densities $\bP$, removing the constraint that $\rho_\bP=\1_\Omega$, we get the \emph{Jellium} energy that was studied in~\cite{LieNar-75}. There it is shown that it is $\geq -(3/5) ({9\pi}/2)^{1/3}|\Omega|$ for any set $\Omega$, hence~\eqref{eq:lower_bd_UEG} follows. 
\end{proof}

For getting upper bounds one can use test functions and numerical calculations. In~\cite{SeiVucGor-16}, a numerical trial state was constructed, giving the numerical upper bound 
\begin{equation}
\frac{E(B)}{|B|}\leq -1.3427 \qquad\text{ for $s=1$ and $d=3$,}
\label{eq:numerics_ball}
\end{equation}
for a ball $B$ of volume $|B|=60$, using tools from optimal transport. It seems reasonable to expect that for a fixed domain $\Omega$ with $|\Omega|=1$, $N\mapsto E(N^{1/3}\Omega)/N$ is decreasing. This is so far an open problem. If valid, then~\eqref{eq:numerics_ball} would imply $e_{\rm UEG}\leq -1.3427$. 

\medskip

For domains $\Omega$ which can be used to tile the space $\R^d$, we can prove like for cubes that $E(\Omega)/|\Omega|$ is always an upper bound to its limit $e_{\rm UEG}$. 

\begin{theorem}[Bound for specific sets]\label{thm:lower_bound_sets}
Let $\Omega$ be a parallelepiped, a tetrahedron or any other convex polyhedron that generates a tiling of $\R^d$. Assume also that $|\Omega|$ is an integer. Then 
$$\frac{E(\Omega)}{|\Omega|}\geq e_{\rm UEG}.$$
\end{theorem}

\begin{proof}
Let $\Omega$ be an open convex set that generates a tiling of $\R^d$. That is, we assume that there exists a discrete subgroup $\Gamma$ of the group of translations and rotations $\R^d\rtimes SO(d)$ such that $\R^d=\cup_{g\in\Gamma} \,g\cdot \overline{\Omega}$ and $g\cdot \Omega\cap g'\cdot\Omega=\emptyset$ for $g\neq  g'$. Note that such a domain $\Omega$ is necessarily a polyhedron. Let then 
$$A_N:=\bigcup_{g\Omega\cap B(0,N^{1/d})\neq \emptyset} g\cdot \overline{\Omega}$$
be the union of all the tiles that intersects the large ball $B(0,N^{1/d})$. By~\cite[Prop.~2]{HaiLewSol_1-09}, $A_N$ satisfies the Fisher regularity condition. By subadditivity, we have
$$\frac{E(A_N)}{|A_N|}\leq \frac{E(\Omega)}{|\Omega|}.$$
Taking $N\to\ii$ gives the result.
\end{proof}

\section{The Grand-Canonical Uniform Electron Gas}\label{sec:grand-canonical}

It will be useful to let the number of particles $N$ fluctuate, and in particular to allow sets $\Omega$ which have a non-integer volume. In this section we define a grand-canonical version of the Uniform Electron Gas and prove that it has the same thermodynamic limit $e_{\rm UEG}$. 

A grand-canonical probability is for us a collection $(\bP_n)_{n\geq0}$ of symmetric probability measures on $(\R^d)^n$ and coefficients $\lambda_n\geq0$ such that $\sum_{n\geq0}\lambda_n=1$. Each $\lambda_n$ gives the probability to have $n$ particles whereas $\bP_n$ gives the precise probability distribution of these $n$ particles. 
The interaction energy of the grand-canonical probability $\bP=\oplus_{n\geq0}\lambda_n\bP_n$ is just the sum of the energies of each component:
$$C(\bP)=\sum_{n=2}^\ii\lambda_n C(\bP_n).$$
Similarly, its density is, by definition,
$$\rho_\bP=\sum_{n=1}^\ii \lambda_n\,\rho_{\bP_n}.$$
In particular, the average number of particles in the system is
$$\int_{\R^d}\rho_\bP=\sum_{n=1}^\ii n\,\lambda_n.$$
We may then define the grand-canonical indirect energy by
\begin{equation}
E_{\rm GC}(\rho):=\inf_{\substack{\bP=\oplus_{n\geq0}\lambda_n\bP_n\\ 
\text{grand-canonical}\\ \text{probability}\\\rho_\bP=\rho}}\bigg\{\sum_{n\geq2}\lambda_n C(\bP_n)\bigg\} - D(\rho,\rho).
\label{eq:def_E_GC}
\end{equation}
When $\int_{\R^d}\rho=N$ is an integer, we have $E_{\rm GC}(\rho)\leq E(\rho)$ since we can restrict the infimum to canonical $N$-particle probabilities $\bP_N$. 

The grand-canonical problem has a very similar structure to the canonical one and we will not give all the arguments again. Of particular importance is the subadditivity
\begin{equation}
 E_{\rm GC}(\rho+\rho')\leq E_{\rm GC}(\rho)+E_{\rm GC}(\rho')
 \label{eq:subadditivity_GC}
\end{equation}
which is now valid for every integrable $\rho,\rho'\geq0$ such that $\rho,\rho'\in L^1(\R^d)\cap L^{1+s/d}(\R^d)$. This inequality can be proved by using the `grand-canonical' tensor product
$$\bP\otimes \bP':=\bigoplus_{n\geq0}\left(\sum_{k=0}^n\lambda_k\lambda'_{n-k}\bP_k\otimes\bP_{n-k}'\right)$$
which has the density $\rho_{\bP\otimes \bP'}=\rho_\bP+\rho_{\bP'}$. In addition, we remark that the Lieb-Oxford inequality 
\begin{equation}
E_{\rm GC}(\rho)\geq -C'_{\rm LO}(s,d)\int_{\R^d}\rho(x)^{1+\frac{s}{d}} \, dx
 \label{eq:Lieb-Oxford_GC}
\end{equation}
is valid in the grand-canonical setting (possibly with a different constant), as can be verified from the proof in~\cite{LieOxf-80,LieSei-09,Bach-92,GraSol-94,LieSolYng-95}. 

Now we look at the grand canonical Uniform Electron Gas. The next result says that the thermodynamic limit is exactly the same as in the canonical case.

\begin{theorem}[Grand Canonical Uniform Electron Gas]\label{thm:UEG_thermo_limit_GC}
Let $\rho_0>0$ and $\{\Omega_N\}\subset\R^d$ be a sequence of bounded connected domains such that 
\begin{itemize}
 \item $|\Omega_N|\to\ii$;
 \item $\Omega_N$ has an $\eta$--regular boundary for all $N$, for some $\eta$ which is independent of $N$.
\end{itemize}
Then
\begin{equation}
 \lim_{N\to\ii}\frac{E_{\rm GC}(\rho_0\1_{\Omega_N})}{|\Omega_N|}=\rho_0^{1+\frac{s}{d}}\,e_{\rm UEG} 
 \label{eq:thermodynamic_limit_GC}
\end{equation}
where $e_{\rm UEG}$ is the same constant as in Theorem~\ref{thm:UEG_thermo_limit}.
\end{theorem}

\begin{proof}
We split our proof into several steps

\medskip

\noindent\textbf{Step 1. Thermodynamic limit of the grand canonical UEG.}
By following step by step the arguments given in the canonical case, we can prove using subadditivity that the thermodynamic limit exists and does not depend on the sequence of domains. In addition, the limit $e_{\rm UEG}^{\rm GC}$ is clearly lower than $e_{\rm UEG}$, as can be verified using sequences for which $\rho_0|\Omega_N|$ is an integer. So the only thing that we have to do is to show that $e^{\rm GC}_{\rm UEG}\geq e_{\rm UEG}$. Our argument is general and could be of independent interest.

\medskip

Our goal will now be to show that 
\begin{equation}
\frac{E_{\rm GC}(C)}{|C|}\geq e_{\rm UEG} 
\label{eq:cube_GC_C}
\end{equation}
for any cube $C$. Using the thermodynamic limit for cubes we will immediately obtain the claimed inequality $e^{\rm GC}_{\rm UEG}\geq e_{\rm UEG}$. To this end we start by giving a lower bound for grand-canonical states which have rational coefficients $\lambda_n$. 

\medskip

\noindent\textbf{Step 2. Construction of a canonical state and a lower bound.} Our main result is the following lemma

\begin{lemma}[Comparing the grand-canonical and canonical energies]\label{lem:compare_GC_C}
Let $\bP=\oplus_{n\geq0}\frac{p_n}{q} \bP_n$ be a grand-canonical probability such that $\rho_{\bP}=\1_C$, where $|C|$, $p_n$ and $q$ are all integers. Let $C'=\cup_{k=1}^qC_k$ be the union of $q$ disjoint copies $C_k$ of $C$. Then 
\begin{equation}
 \frac{C(\bP)-D(\1_C,\1_C)}{|C|}\geq \frac{E(C')}{|C'|}-\frac{D\big(\1_{C'},\1_{C'}\big)}{(q-1)|C'|}
 \label{eq:estim_large_cube_GC}
\end{equation}
where $E(C')$ is the canonical indirect energy defined in~\eqref{eq:def_E}.
\end{lemma}

\begin{proof}[Proof of Lemma~\ref{lem:compare_GC_C}]
It is more convenient to write $\bP=q^{-1}\sum_{j=1}^q\bQ_j$ where the $\bQ_j$ are equal to the $\bP_n$ (for which $p_n\neq0$), each of them being repeated $p_n$ times. In our argument, we do not need to know the exact value of the number of particles of $\bQ_j$, which we call $n_j$.  
Note that $\sum_{j=1}^q n_j=Nq$ where $N=|C|$.

Then we consider $q$ disjoint copies of the cube $C$ as in the statement, which we call $C_1,...,C_q$ and we build a particular canonical state with $Nq$ particles, living in the union $C'=\cup_{j=1}^qC_j$. Let us denote by $\bQ_j^{C_k}$ the probability measure $\bQ_j$ placed into the cube $C_k$. Then our canonical state consists in placing each of the states $\bQ_j$ in one of the $q$ cubes $C_k$ and then looking at all the possible permutations:
$$\bQ:=\frac{1}{q!}\sum_{\sigma\in\gS_q}\bQ_{\sigma(1)}^{C_1}\otimes \bQ_{\sigma(2)}^{C_2}\otimes\cdots\otimes \bQ_{\sigma(q)}^{C_q}.$$
It can be verified that the restriction of $\bQ$ into each cube $C_j$ is precisely $\bP$, shifted into that cube.\footnote{See Step 3 of the proof of Theorem~\ref{thm:rho_scaled} below for a precise definition of the restriction.} In other words, $\bQ$ is a canonical state living over the big set $\cup_{k=1}^qC_k$ such that each local restriction is equal to the original grand canonical measure. In particular, the density of $\bQ$ is
$$\rho_{\bQ}=\1_{C'}=\sum_{k=1}^q\1_{C_k}.$$
In the following we also denote by $\rho_j^{C_k}$ the density of $\bQ_j^{C_k}$, hence
$$\sum_{j=1}^q\rho_j^{C_k}=q\1_{C_k}.$$
The Coulomb energy of $\bQ$ is, similarly as in the proof of Lemma~\ref{lem:subadditivity},
\begin{align*}
C(\bQ)=& \frac{1}{q!}\sum_{\sigma\in\gS_q}\left(\sum_{j=1}^qC(\bQ_j)+\sum_{j\neq k}D\big(\rho_{\sigma(j)}^{C_j},\rho_{\sigma(k)}^{C_k}\big)\right)\\
=&q\,C(\bP)+\frac{1}{q(q-1)}\sum_{j\neq k}\sum_{\ell\neq m}D\big(\rho_{\ell}^{C_j},\rho_{m}^{C_k}\big)\\
=&q\,C(\bP)+\frac{q}{q-1}\sum_{j\neq k}D\big(\1_{C_j},\1_{C_k}\big)-\frac{1}{q(q-1)}\sum_{j\neq k}\sum_{\ell}D\big(\rho_{\ell}^{C_j},\rho_{\ell}^{C_k}\big)\\
=&q\,C(\bP)+\frac{q}{q-1}D\big(\rho_\bQ,\rho_\bQ\big)-\frac{q^2}{q-1}D\big(\1_{C},\1_{C}\big)\\
&\qquad -\frac{1}{q(q-1)}\sum_{j\neq k}\sum_{\ell}D\big(\rho_{\ell}^{C_j},\rho_{\ell}^{C_k}\big)\\
\leq&q\,\Big(C(\bP)-D(\1_C,\1_C)\Big)+\frac{q}{q-1}D\big(\rho_\bQ,\rho_\bQ\big).
\end{align*}
Therefore we have shown that
$$C(\bQ)-D\big(\rho_\bQ,\rho_\bQ\big)\leq q\,\Big(C(\bP)-D(\1_C,\1_C)\Big)+\frac{1}{q-1}D\big(\rho_\bQ,\rho_\bQ\big).$$
Dividing by $q|C|$ gives~\eqref{eq:estim_large_cube_GC}.
\end{proof}

\medskip

\noindent\textbf{Step 3. Proof of the lower bound~\eqref{eq:cube_GC_C} for any cube.}

\begin{lemma}[Canonical lower bound for cubes]
Let $C$ be any cube of integer volume. Then
\begin{equation}
 \frac{E_{\rm GC}(C)}{|C|}\geq e_{\rm UEG}
\label{eq:cube_GC_C_lem}
\end{equation}
where $e_{\rm UEG}$ is the canonical energy defined in Theorem~\ref{thm:UEG_thermo_limit}.
\end{lemma}

\begin{proof}
For a probability of the special form $\bP=\oplus_{n\geq0}\frac{p_n}{q} \bP_n$, we can arbitrarily increase $q$ while keeping $\bP$ fixed by multiplying $p_n$ and $q$ by the same number $k$. The set $C'_k$ in Lemma~\ref{lem:compare_GC_C} is the union of $qk$ cubes which we can pack such as to form a domain of diameter proportional to $(qk)^{1/d}|C|^{1/d}$. By Theorem~\ref{thm:UEG_thermo_limit} we have as $k\to\ii$
$$\lim_{k\to\ii}\frac{E(C'_k)}{|C'_k|}=e_{\rm UEG}.$$
On the other hand we have
$$\frac{D\big(\1_{C'_k},\1_{C'_k}\big)}{(qk-1)|C'_k|}=O\left(\frac{|C'_k|^{1-\frac{s}d}}{qk}\right)=O\left(\frac{|C|^{1-\frac{s}d}}{(qk)^{\frac{s}d}}\right)$$
which tends to 0 as $k\to\ii$. Thus~\eqref{eq:estim_large_cube_GC} in Lemma~\ref{lem:compare_GC_C} shows that
\begin{equation}
 \frac{C(\bP)-D(\1_C,\1_C)}{|C|}\geq e_{\rm UEG}
 \label{eq:lower_bound_fractional_state_GC}
\end{equation}
for any $\bP=\oplus_{n\geq0}\frac{p_n}{q} \bP_n$ and any cube $C$ of integer volume $N$.

Next we use a density argument to deduce the same property for any $\bP=\oplus_{n\geq0}\lambda_n\bP_n$. For any fixed $\eps$ and $M$, we can find $p_n$ and $q$ such that $\lambda_n-\eps\leq p_n/q\leq \lambda_n$, for $n=0,...,M$. Let
$$\rho_{\eps,M}=\sum_{n=0}^M\left(\lambda_n-\frac{p_n}{q}\right)\rho_{\bP_n}+\sum_{n\geq M+1}\lambda_n\rho_{\bP_n}=\1_C-\sum_{n=0}^M\frac{p_n}{q}\rho_{\bP_n}$$
and note that
$$\int_{\R^d}\rho_{\eps,M}=N-\sum_{n=0}^Mn\frac{p_n}{q}$$ 
is a rational number. Of course, $\int_{\R^d}\rho_{\eps,M}\to0$ when $\eps\to0$ and $M\to\ii$.
Define then the tensor product
$$\bP_{\eps,M}=\bP_{\eps,M}^{(1)}\otimes \bP_{\eps,M}^{(2)}$$
with
$$\bP_{\eps,M}^{(1)}=\left(1-\sum_{n=1}^M\frac{p_n}{q}\right)\oplus\bigoplus_{n=1}^M\frac{p_n}{q}\bP_n\oplus0$$
and
$$\bP_{\eps,M}^{(2)}= \left(1-\int_{\R^d}\rho_{\eps,M}\right)\oplus\left(\int_{\R^d}\rho_{\eps,M}\right)\frac{\rho_{\eps,M}}{\int_{\R^d}\rho_{\eps,M}}\oplus0.$$
The probability $\bP_{\eps,M}$ has only rational coefficients and at most $M+1$ particles. Its density is
$$\rho_{\bP_{\eps,M}}=\rho_{\bP_{\eps,M}^{(1)}}+\rho_{\bP_{\eps,M}^{(2)}}=\sum_{n=1}^M\frac{p_n}{q}\rho_{\bP_n}+\rho_{\eps,M}=\1_C.$$
Thus by~\eqref{eq:lower_bound_fractional_state_GC} we have
$$\frac{C(\bP_{\eps,M})-D(\1_C,\1_C)}{|C|}\geq e_{\rm UEG}.$$
Passing to the limit $\eps\to0$ and $M\to\ii$ we deduce that 
$$\frac{C(\bP)-D(\1_C,\1_C)}{|C|}\geq e_{\rm UEG},$$
as we wanted.
\end{proof}

This completes the proof of Theorem~\ref{thm:UEG_thermo_limit_GC}.
\end{proof}

Repeating the above proof for a tile different from a cube, we can obtain the following result.

\begin{corollary}\label{cor:tetrahedra_GC}
Let $\Omega$ be a parallelepiped, a tetrahedron or any other convex polyhedron that generates a tiling of $\R^d$. Assume also that $|\Omega|$ is an integer. Then
\begin{equation}
 \frac{E_{\rm GC}(\Omega)}{|\Omega|}\geq e_{\rm UEG}
\label{eq:cube_GC_tetrahedron_lem}
\end{equation}
where $e_{\rm UEG}$ is the canonical energy defined in Theorem~\ref{thm:UEG_thermo_limit}.
\end{corollary}

\section{Limit for slowly varying densities}\label{sec:limit_slowly_varying_rho}

In this section we look at the special case of a slowly varying density, in the Coulomb case
$$s=1\quad \text{and} \quad d=3.$$
Namely, we take $\rho_N(x)=\rho(x/N^{1/3})$ for a given $\rho$ with $\int_{\R^3}\rho\in\N$ and we prove that the limit $N\to\ii$ of the corresponding indirect energy is the uniform electron gas energy. This type of scaled density was used in several recent computations~\cite{SeiGorSav-07,RasSeiGor-11,SeiVucGor-16}. 

\begin{theorem}[Limit for scaled densities]\label{thm:rho_scaled}
Take $s=1$ in dimension $d=3$. Let $\rho\geq0$ be any continuous density on $\R^3$ such that $\int_{\R^3}\rho\in\mathbb{N}$ and $\rho\in\ell^1(L^\ii)$, which means that 
$$\sum_{k\in \Z^3}\max_{k+[0,1)^3}\rho<\ii.$$
Then we have
\begin{equation}
\boxed{\lim_{N\to\ii} \frac{E\big(\rho(x/N^{1/3})\big)}{N}= \lim_{N\to\ii} \frac{E\big(N\rho(x)\big)}{N^{4/3}}=e_{\rm UEG}\int_{\R^3}\rho(x)^{4/3}\,dx}
\label{eq:limit_slowly_varying_rho}
\end{equation}
where $e_{\rm UEG}$ is the constant defined in Theorem~\ref{thm:UEG_thermo_limit}.
\end{theorem}

The intuition behind the theorem is that $\rho_N(x)=\rho(x/N^{1/3})$ becomes almost constant locally since its derivative behaves as $N^{-1/3}$. Although the indirect energy is not local, the correlations are weak for slowly varying densities and the limit is the ``local density approximation'' of the indirect energy.

It has been numerically observed that $N\mapsto E(\rho_N)/N$ is decreasing for many choices of $\rho$~\cite{SeiGorSav-07,RasSeiGor-11,SeiVucGor-16}. If we could prove that for any $\rho$, $N\mapsto E(\rho_N)/N$ is indeed decreasing, then we would immediately conclude that the best Lieb-Oxford constant is $-e_{\rm UEG}$, settling thereby a longstanding conjecture.

Our assumption that $\rho$ is continuous can be weakened, for instance by only requiring that $\rho$ is piecewise continuous (with smooth discontinuity surfaces). Our proof requires to be able to approximate $\rho$ from below and above by step functions, and we therefore cannot treat an arbitrary function in $L^1\cap L^{4/3}$. We essentially need that $\rho$ is Riemann-integrable.

The result of Theorem~\ref{thm:rho_scaled} should be compared with recent works of Sandier, Serfaty and co-workers~\cite{SanSer-15,Serfaty-14,SanSer-14a,RouSer-16,RotSer-15,PetSer-15} on the first-order correction to the mean-field limit, for any $\max(0,d-2)\leq s<d$ in dimension $d\geq1$. In those works the potential $V$ is fixed to be $V(x/N^{1/s})$ (or equivalently $NV(x)$) and there is no constraint on the density. The main result of those works is that 
\begin{multline}
\inf_{x_1,...,x_N\in \R^d}\bigg\{\sum_{j=1}^NV\left(\frac{x_j}{N^{1/s}}\right)+\sum_{1\leq j<k\leq N}\frac{1}{|x_j-x_k|^s}\bigg\}\\
=N\min_{\substack{\rho\geq0\\ \int_{\R^d}\rho=1}}\bigg\{\int_{\R^d}V\rho+D(\rho,\rho)\bigg\}+N^{\frac{s}{d}}\,e_{\rm Jel}\int_{\R^d}\overline\rho(x)^{1+\frac{s}{d}}\,dx+o\left(N^{\frac{s}{d}}\right),
\label{eq:mean_field_Sandier_Serfaty}
\end{multline}
where $\overline\rho$ is the unique minimizer to the minimum on the right and $e_{\rm Jel}\leq e_{\rm UEG}$ is the Jellium energy. We see from~\eqref{eq:limit_slowly_varying_rho} and~\eqref{eq:mean_field_Sandier_Serfaty} that the Jellium model arises when the potential is fixed (and the density is optimized), whereas the UEG arises when the density is fixed. Whether $e_{\rm UEG}$ is equal to $e_{\rm Jel}$ or not is an important question in DFT.

In order to allow for a better comparison, it would be interesting to extend our limit~\eqref{eq:limit_slowly_varying_rho} to all $0<s<d$ in any dimension (in dimension $d=1$, this has recently been done in~\cite{DiMarino-17}).

\begin{proof}
As usual we have to prove a lower and an upper bound.

\medskip

\noindent\textbf{Step 1. A simple approximation lemma.} 
The following is an elementary result about the approximation of continuous functions by step functions.

\begin{lemma}\label{lem:approx_step_fn}
Let $f\geq0$ be a continuous function in $\ell^1(L^\ii)$. For every $\epsilon>0$, consider a tiling of the full space $\R^3=\cup_j \overline{D_j}$ with pairwise disjoint polyhedral domains such that ${\rm diam}(D_j)\leq\epsilon$. Define the approximations
\begin{equation}
 f_\epsilon^-(x):=\sum_j \big(\min_{D_j}f\big)\1_{D_j}(x),\qquad f_\epsilon^+(x):=\sum_j \big(\max_{D_j}f\big)\1_{D_j}(x). 
 \label{eq:limit_tiling}
\end{equation}
Then $f_\epsilon^\pm\to f$ strongly in $L^1(\R^3)\cap L^\ii(\R^3)$. 

Similarly, we have for any fixed open bounded set $\Omega$ and any $1\leq p<\ii$
\begin{equation}
\lim_{\eps\to0}\int_{\R^3}\left(\min_{x+\eps\Omega}f\right)^pdx=\lim_{\eps\to0}\int_{\R^3}\left(\max_{x+\eps\Omega}f\right)^pdx=\int_{\R^3}f(x)^p\,dx.
\label{eq:limit_continuous_tiling}
\end{equation}
\end{lemma}

The limit~\eqref{eq:limit_continuous_tiling} is similar to~\eqref{eq:limit_tiling} in that it can be interpreted as a kind of continuous tiling with the small domain $\eps\Omega$.

\begin{proof}
Since $f\in \ell^1(L^\ii)$ then it must tend to 0 at infinity and it is therefore uniformly continuous on $\R^3$. This implies that $f_\epsilon^\pm\to f$ uniformly on $\R^3$. Then we have
\begin{align*}
\int f_\eps^+-f^-_\eps&=\sum_j |D_j|\Big(\max_{D_j}f-\min_{D_j}f\Big)\\
&=\sum_{k\in\Z^3}\underbrace{\sum_{D_j\cap (k+[0,1)^3)\neq\emptyset} |D_j|\Big(\max_{D_j}f-\min_{D_j}f\Big)}_{\leq 2\|f\|_{L^\ii(k+[0,1)^3)}}. 
\end{align*}
This tends to zero, by the dominated convergence theorem. The argument is similar for~\eqref{eq:limit_continuous_tiling}.
\end{proof}

There is a similar result when $f$ is only piecewise continuous (without the uniform convergence).

\medskip

\noindent\textbf{Step 2. Upper bound.} 
Let $\ell$ be an integer such that $1\ll\ell\ll N^{1/3}$ and $\R^3=\cup_j C_j$ be a tiling made of cubes of side length $\ell$. In each cube $C_j$, let 
$$m_j=\min_{C_j} \rho_N$$
be the minimum value of $\rho_N(x)=\rho(x/N^{1/3})$. Take also
$$\eps_j=1-\frac{\lfloor m_j|C_j|\rfloor}{m_j|C_j|}\leq \min\left(1,\frac{1}{m_j|C_j|}\right)$$
such that $m_j(1-\epsilon_j)|C_j|=\lfloor m_j|C_j|\rfloor$ is an integer.
Then
$$\rho_N\geq \sum_jm_j(1-\varepsilon_j)\1_{C_j}.$$
By subadditivity and using $(1-\eps)^{4/3}\geq 1-4\eps/3$, we obtain 
\begin{align*}
E(\rho_N)\leq &\sum_j E\big(m_j(1-\eps_j)\1_{C_j}\big)\\
=&\sum_j m_j^{\frac43}(1-\eps_j)^{\frac43}|C_j|\, \frac{E(C'_j)}{|C'_j|}\\
\leq&\sum_j (1-\eps_j)^{4/3}m_j^{\frac43}|C_j|\left( \frac{E(C'_j)}{|C'_j|}-e_{\rm UEG}\right)+e_{\rm UEG}\sum_j (1-\eps_j)^{4/3}m_j^{\frac43}|C_j|\\
\leq&\sum_j (1-\eps_j)^{4/3}m_j^{\frac43}|C_j|\left( \frac{E(C'_j)}{|C'_j|}-e_{\rm UEG}\right)-\frac43 e_{\rm UEG}\sum_j m_j^{\frac43}\eps_j|C_j|\\
&+N\,e_{\rm UEG}\int_{\R^3}\rho^{\frac43}+e_{\rm UEG}\left(\sum_j m_j^{\frac43}|C_j|-\int_{\R^3}\rho_N^{\frac43}\right),
\end{align*}
where $C'_j=m_j^{1/3}(1-\eps_j)^{1/3}C_j$. 

We estimate the error terms as follows. After scaling by $N^{1/3}$ we find by Lemma~\ref{lem:approx_step_fn}
\begin{equation}
\int_{\R^3}\rho_N^{\frac43}-\sum_j m_j^{\frac43}|C_j| =N\int_{\R^3}\left(\rho^{\frac43}-(\rho_{\ell/N^{1/3}}^-)^{\frac43}\right)=o(N)
\label{eq:scaling_N}
\end{equation}
where $\rho_{\ell/N^{1/3}}^-$ is defined as in~\eqref{eq:limit_tiling} with the tiling $D_j=N^{-1/3}C_j$.
For the second error term we write\footnote{Here and everywhere else, $c>0$ denotes a constant that may change from line to line.}
\begin{align*}
\sum_j \eps_j m_j^{\frac43}|C_j|&\leq  c\eta\sum_{\eps_j\leq\eta} m_j^{\frac43}|C_j|+\sum_{\eps_j\geq\eta} m_j^{\frac43}|C_j|\\
 &\leq  c\eta\sum_{\eps_j\leq\eta} m_j^{\frac43}|C_j|+\frac{1}{\ell\eta^{1/3}}\sum_{\eps_j\geq\eta} m_j|C_j|,
\end{align*}
since $\eps_j\geq\eta$ implies $m_j\leq \ell^{-3}\eta^{-1}$. Using again Lemma~\ref{lem:approx_step_fn} as in~\eqref{eq:scaling_N} gives that the two sums grow linearly in $N$, hence
$$\sum_j \eps_jm_j^{\frac43}|C_j|\leq cN\left(\eta+\frac{1}{\ell\eta^{1/3}}\right)\leq cN\ell^{-\frac34}.$$
Similarly, since $E(C'_j)|C'_j|^{-1}$ is uniformly bounded by~\eqref{eq:LO}, we have
\begin{align*}
&\sum_j m_j^{\frac43}(1-\eps_j)^{4/3}|C_j|\left( \frac{E(C'_j)}{|C'_j|}-e_{\rm UEG}\right)\\
&\quad \leq c\sum_{m_j(1-\epsilon_j)\leq\eta} m_j^{\frac43}(1-\eps_j)^{4/3}|C_j|+\sum_{m_j(1-\epsilon_j)\geq\eta} m_j^{\frac43}|C_j|\left( \frac{E(C'_j)}{|C'_j|}-e_{\rm UEG}\right)\\
&\quad \leq cN\eta^{1/3}+\sum_{m_j(1-\epsilon_j)\geq\eta} m_j^{\frac43}|C_j|\left( \frac{E(C'_j)}{|C'_j|}-e_{\rm UEG}\right).
\end{align*}
In the second sum, $|C'_j|=m_j(1-\eps_j)\ell^3\geq \ell^3\eta$. So 
$$\frac{E(C'_j)}{|C'_j|}\to e_{\rm UEG}$$
as long as $\eta$ is chosen such that $\ell^3\eta\to\ii$. From the dominated convergence theorem and~\eqref{eq:scaling_N}, it follows that 
$$\sum_{m_j(1-\epsilon_j)\geq\eta} m_j^{\frac43}|C_j|\left( \frac{E(C'_j)}{|C'_j|}-e_{\rm UEG}\right)=o(N).$$
So taking $\eta\to0$ with $\eta\ell^3\to\ii$, we have proved that
$$E(\rho_N)\leq N\,e_{\rm UEG}\int_{\R^3}\rho^{\frac43}+o(N).$$

\medskip

\noindent\textbf{Step 3. Lower bound.} So far our argument was very general and works exactly the same for any $0<s<d$. For the lower bound we use the Graf-Schenker inequality~\cite{GraSch-95,HaiLewSol_1-09,HaiLewSol_2-09}, which enables to decouple Coulomb subsystems using a tiling made of tetrahedra and averaging over translations and rotations of the tiling. This inequality is very specific to the 3D Coulomb case and it is a powerful tool to use screening effects.

In order to go further, we need the concept of localized classical states~\cite{Lewin-11,FouLewSol-15}. If we have a canonical symmetric $N$-particle density $\bP$, we define its localization $\bP_{|A}$ to a set $A$ by the requirement that all its $k$-particle densities are equal to $\rho^{(k)}_{\bP_{|A}}=(1_A)^{\otimes k}\rho^{(k)}_{\bP}$, namely, those are localized in the usual way. Except when all the $N$ particles are always inside or outside of $A$, the localized state $\bP_{|A}$ must be a grand-canonical state, since the number of particles in $A$ fluctuates. More precisely, $\bP_{|A}$ is the sum of the $n+1$ probabilities defined by
\begin{equation}
 \begin{cases}
  \bP_{|A,0}=\int_{(\R^3\setminus A)^N}d\bP,\\
  \bP_{|A,n}(x_1,...,x_n)=\1_{A^{\otimes n}}(x_1,...,x_n){N\choose n}\int_{(\R^3\setminus A)^{N-n}}d\bP(x_1,...,x_n,\cdot).
 \end{cases}
\label{eq:def_localized}
\end{equation}
That the localization of a canonical state is always a grand-canonical state is our main motivation for having considered the grand-canonical UEG in Section~\ref{sec:grand-canonical}. It is actually possible to define the localization $\bP_\chi$ for any function $|\chi|\leq1$ and not only for characteristic functions. It suffices to replace everywhere $\1_A$ by $\chi^2$ and $\1_{\R^3\setminus A}$ by $1-\chi^2$. This will be used later in the quantum case, where smooth localization functions are mandatory.

In our setting, the Graf-Schenker inequality says that the full indirect energy can be bounded from below by the average of the energies of the localized states in a tetrahedron, which is rotated in all directions and translated over the whole of $\R^3$. This is the same as taking a tiling made of simplices and averaging over translations and rotations of this tiling. 

\begin{lemma}[Graf-Schenker inequality for the exchange-correlation energy]\label{lem:Graf-Schenker}
Let $s=1$ and $d=3$. Let $\bDelta\subset\R^3$ be a tetrahedron. There exists a constant $c>0$ such that for every $\ell>0$ and every $N$-particle symmetric probability $\bP$, we have
\begin{multline}
C(\bP)-D(\rho_{\bP},\rho_{\bP})\\
\geq \frac1{|\ell\bDelta|}\int_{\R^3\times SO(3)}\Big\{C\big(\bP_{|g\ell\bDelta}\big)-D(\rho_{\bP}\1_{g\ell\bDelta},\rho_\bP\1_{g\ell\bDelta})\Big\}\,dg-\frac{c}{\ell}\int_{\R^3}\rho_\bP
\label{eq:lower_bd_Graf_Schenker}
\end{multline}
where $\bP_{|g\ell\bDelta}$ is the (grand-canonical) restriction of $\bP$ to the subset $g\ell\bDelta$.
\end{lemma}

\begin{proof}
Graf and Schenker have proved in~\cite{GraSch-95} that the potential 
$$w_\ell(x)=\frac{1-h_\ell(x)}{|x|}$$
has positive Fourier transform, where 
\begin{align*}
h_\ell(x-y)&=\frac1{|\ell\bDelta|}\int_{SO(3)}\1_{\ell\bDelta}\ast\1_{-\ell\bDelta}(Rx-Ry)\,dR\\
&=\frac1{|\ell\bDelta|}\int_{SO(3)}\int_{\R^3}\1_{R^{-1}\ell \bDelta+z}(y)\1_{R^{-1}\ell \bDelta+z}(x)\,dz\,dR=h_1\left(\frac{x-y}{\ell}\right) 
\end{align*}
and with $\bDelta$ a tetrahedron. Note that $h_\ell(0)=1$ and that $h_\ell'(0)=h_1'(0)/\ell$. 
From this we deduce that for any $\rho$,
\begin{align}
0&\leq \frac12\int_{\R^3}\int_{\R^3}w_\ell(x-y)\left(\sum_{j=1}^N\delta_{x_j}(x)-\rho(x)\right)\left(\sum_{j=1}^N\delta_{x_j}(y)-\rho(y)\right)\,dx\,dy \nn \\
&=\sum_{1\leq j<k\leq N}w_\ell(x_j-x_k)+\frac{Nh_1'(0)}{2\ell}-2\sum_{j=1}^ND_{w_\ell}(\rho,\delta_{x_j})+D_{w_\ell}(\rho,\rho),\label{eq:Graf_Schenker_use}
\end{align}
since $D_{w_\ell}(f,g)=1/2\int_{\R^3}\int_{\R^3}w_\ell(x-y)f(x)g(y)\,dx\,dy$ is positive-definite.
Taking $\rho=\rho_{\bP}$ and integrating against $\bP$, we get 
$$\int_{\R^{3N}}\sum_{1\leq j<k\leq N}w_\ell(x_j-x_k)d\mathbb{P}(x_1,...,x_N)\geq D_{w_\ell}(\rho_{\bP},\rho_{\bP})-\frac{cN}{\ell}.$$
Inserting the definition of $w_\ell$, this can be stated in the form~\eqref{eq:lower_bd_Graf_Schenker}.
\end{proof}

Applying~\eqref{eq:lower_bd_Graf_Schenker} to a probability $\bP$ such that $\rho_\bP=\rho_N=\rho(\cdot/N^{1/3})$, we find 
\begin{align*}
E(\rho_N)\geq \frac1{|\ell\bDelta|}\int_{\R^3\times SO(3)}E_{\rm GC}(\rho_N\1_{g\ell\bDelta})\,dg-c\frac{N}{\ell}.
\end{align*}
For each tetrahedron $g\ell\bDelta$ we denote by $M(g,\ell)=\max_{g\ell\bDelta} \rho_N$ and we take $\eps(g,\ell)\in [0,\min(1,M(g,\ell)^{-1}|\ell\bDelta|^{-1}))$ to ensure that $(1+\eps(g,\ell))M(g,\ell)|\ell\bDelta|$ is the next integer after $M(g,\ell)|\ell\bDelta|$. Then we have by Corollary~\ref{cor:tetrahedra_GC}
\begin{align*}
E_{\rm GC}(\rho_N\1_{g\ell\bDelta})&\geq E_{\rm GC}\big((1+\eps(g,\ell))M(g,\ell)\1_{g\ell\bDelta}\big)\\
&\geq (1+\eps(g,\ell))^{4/3}M(g,\ell)^{4/3}|\ell\bDelta|\, e_{\rm UEG}\\
&\geq (1+2\eps(g,\ell))M(g,\ell)^{4/3}|\ell\bDelta|\, e_{\rm UEG}.
\end{align*}
Thus
$$E(\rho_N)\geq e_{\rm UEG}\int_{\R^3\times SO(3)}(1+2\eps(g,\ell))\max_{g\ell\bDelta}\rho_N^{4/3}\,dg-c\frac{N}{\ell}.$$
We have
\begin{align*}
\frac1N\int_{\R^3\times SO(3)}\left(\max_{g\ell\bDelta}\rho_N^{4/3}\right)\,dg&=\int_{\R^3\times SO(3)}\left(\max_{g\ell N^{-1/3}\bDelta}\rho^{4/3}\right)\,dg\\
&=\int_{\R^3}\int_{SO(3)}\left(\max_{R\ell N^{-1/3}\bDelta+z}\rho^{4/3}\right)\,dR\,dz\\
&\underset{\frac{\ell}{N^{1/3}}\to0}{\longrightarrow}\int_{\R^3}\rho^{4/3}(z)\,dz
\end{align*}
as we want, by the dominated convergence theorem (it is also possible to rewrite the integral over $g$ as an average over translations and rotations of one tiling made of tetrahedra~\cite{GraSch-95,HaiLewSol_1-09,HaiLewSol_2-09} and then to apply Lemma~\ref{lem:approx_step_fn} for this tiling). The term with $\eps(g,\ell)$ is treated as before by writing
\begin{align*}
\int_{\R^3\times SO(3)}\eps(g,\ell)\max_{g\ell\bDelta}\rho_N^{4/3}\,dg&\leq \eta\int_{\eps(g,\ell)\leq\eta}\max_{g\ell\bDelta}\rho_N^{4/3}\,dg+\int_{\eps(g,\ell)\geq\eta}\max_{g\ell\bDelta}\rho_N^{4/3}\,dg\\
&\leq c\eta N+\frac{1}{\ell\eta^{1/3}}\int_{\eps(g,\ell)\geq\eta}\max_{g\ell\bDelta}\rho_N\,dg\leq c\frac{N}{\ell^{3/4}}.
\end{align*}
This concludes the proof of Theorem~\ref{thm:rho_scaled}.
\end{proof} 

\section{Extension to the quantum case}\label{sec:quantum}

In this last section we discuss the quantum case. Of course we cannot employ sharp densities and we have to restrict ourselves to regular densities. A theorem of Harriman~\cite{Harriman-81} and Lieb~\cite{Lieb-83b} says that the set of densities $\rho$ which come from a quantum state with finite kinetic energy is exactly composed of the functions $\rho\geq0$ such that $\sqrt\rho\in H^1(\R^3)$. So we have to work under these assumptions. 

For simplicity we only define the grand canonical UEG, but we expect that the exact same results hold in the canonical setting. We also restrict ourselves to the physical case $s=1$ and $d=3$.

For $\rho\geq0$ with $\sqrt\rho\in H^1(\R^3)$, we define the grand canonical quantum energy by
\begin{multline}
E_{\hbar}(\rho):=\\\inf_{\substack{\Gamma_n=\Gamma_n^*\geq0\\ \Gamma_n\text{ antisymmetric}\\ \sum_{n=0}^\ii\tr(\Gamma_n)=1\\ \sum_{n=1}^\ii \rho_{\Gamma_n}=\rho}} \Bigg\{\sum_{n=1}^\ii \tr_{L^2_a((\R^3\times\{1,...,q\})^n,\C)}\Bigg(-\hbar^2\sum_{j=1}^n\Delta_{x_j}+\sum_{1\leq j<k\leq n}\frac{1}{|x_j-x_k|}\Bigg)\Gamma_n\Bigg\}\\-\frac12\int_{\R^3}\int_{\R^3}\frac{\rho(x)\rho(y)}{|x-y|}dx\,dy.
\label{eq:def_E_quantum}
\end{multline}
Here $L^2_a((\R^3\times \{1,...,q\})^n,\C)$ is the space of antisymmetric square-integrable functions on $(\R^3\times \{1,...,q\})^n$, with $q$ spin states (for electrons $q=2$). The density of $\Gamma_n$ is defined by
\begin{multline*}
\rho_{\Gamma_n}(x)=n\times\\
\times\!\!\!\sum_{\sigma_1,...,\sigma_n\in\{1,...,q\}}\int_{\R^{3(n-1)}}
\Gamma_n(x,\sigma_1,x_2,...,x_n,\sigma_n;x,\sigma_1,x_2,...,x_n,\sigma_n)dx_2\cdots dx_n 
\end{multline*}
where $\Gamma_n(X;Y)$ is the kernel of the trace-class operator $\Gamma_n$. This kernel is such that 
\begin{align*}
&\Gamma_n(x_{\tau(1)},\sigma_{\tau(1)},...,x_{\tau(N)},\sigma_{\tau(N)}\;;\;x'_1,\sigma'_1,...,x'_N,\sigma'_N)\\
&\qquad=\Gamma_n(x_1,\sigma_1,...,x_N,\sigma_N\;;\;x'_{\tau(1)},\sigma'_{\tau(1)}...,x'_{\tau(N)},\sigma'_{\tau(N)})\\
&\qquad =\eps(\tau)\;\Gamma_n(x_1,\sigma_1,...,x_N,\sigma_N\;;\;x'_1,\sigma'_1...,x'_N,\sigma'_N)
\end{align*}
for every permutation $\tau\in\gS^N$ with signature $\eps(\tau)\in\{\pm1\}$.
The exchange-correlation energy is defined in chemistry by subtracting a kinetic energy term, which we do not do here. Hence our energy $E_\hbar(\rho)$ contains all of the kinetic energy for the given $\rho$. 

There are several possibilities to define the quantum uniform electron gas, which should all lead to the same answer. We could for instance work in a domain with Neumann boundary conditions and impose that $\rho$ be exactly constant over this domain. Instead we prefer to impose Dirichlet boundary conditions. More precisely, we ask that $\rho\equiv1$ inside $\Omega$ and that $\rho\equiv0$ outside, where the inside and outside are defined by looking at the points which are at a distance $\ell\ll|\Omega|^{1/3}$ from the boundary $\partial\Omega$, such that the transition region has a negligible volume compared to $|\Omega|$. In the transition region, we only impose that $0\leq\rho\leq1$. Although we expect that the limit will be the same whatever $\rho$ does in this region, it is convenient to look at the worst case, namely, to minimize the energy over all possible such $\rho$.

\begin{theorem}[Quantum Uniform Electron Gas]\label{thm:UEG_thermo_limit_quantum}
Let $\rho_0>0$, $\hbar>0$, $s=1$ and $d=3$. Let $\{\Omega_N\}\subset\R^3$ be a sequence of bounded connected domains such that 
\begin{itemize}
 \item $|\Omega_N|\to\ii$;
 \item $\Omega_N$ has an $\eta$--regular boundary for all $N$, for some $\eta$ which is independent of $N$.
\end{itemize}
Let $\ell_N\to\ii$ be any sequence such that $\ell_N/|\Omega_N|^{1/3}\to0$ and define the inner and outer approximations of $\Omega_N$ by 
$$\Omega_N^-:=\left\{x\in\Omega_N\ :\ {\rm d}(x,\partial\Omega_N)\geq \ell_N\right\},$$
$$\Omega_N^+:=\Omega_N\cup \left\{x\in\R^3\ :\ {\rm d}(x,\partial\Omega_N)\leq \ell_N\right\}.$$
Then the following thermodynamic limit exists
\begin{equation}
 \lim_{N\to\ii}\left(\inf_{\substack{\sqrt{\rho}\in H^1(\R^3)\\ \rho_0\1_{\Omega_N^-}\leq \rho\leq \rho_0\1_{\Omega_N^+} }}\frac{E_\hbar(\rho)}{|\Omega_N|}\right)
=\rho_0^{4/3}\,e_{\rm UEG}\big(\hbar^2\rho_0^{1/3}\big) 
 \label{eq:thermodynamic_limit_quantum}
\end{equation}
where the function $e_{\rm UEG}(\lambda)$ is independent of the sequence $\{\Omega_N\}$ and of $\ell_N$. In addition $\lambda\mapsto e_{\rm UEG}(\lambda)$ is a concave increasing function of $\lambda\in\R^+$ which satisfies
\begin{equation}
 \lim_{\lambda\to0}e_{\rm UEG}(\lambda)=e_{\rm UEG}, 
 \label{eq:limit_quantum_classical_small_rho}
\end{equation}
the classical energy of the UEG defined in Theorem~\ref{thm:UEG_thermo_limit}, and
\begin{equation}
e_{\rm UEG}(\lambda)=c_{\rm TF}\,\lambda-c_{\rm D}+o(1)_{\lambda\to\ii}, 
\label{eq:limit_quantum_classical_large_rho}
\end{equation}
where
$$c_{\rm TF}=\frac3{5}\left(\frac{6\pi^2}{q}\right)^{\tfrac{2}{3}}\quad\text{and}\quad c_{\rm D}=\frac34\left(\frac{6}{q\pi}\right)^{\tfrac{1}{3}}$$
are the Thomas-Fermi and Dirac constants, with $q$ the number of spin states.
\end{theorem}

The main difficulty in the quantum case is that the subadditivity property~\eqref{eq:subadditivity} does not hold anymore. Although we expect a similar property with small error terms, proving it would require to deal with overlapping quantum states, which is not easy. Our proof of Theorem~\ref{thm:UEG_thermo_limit_quantum} will therefore bypass this difficulty and instead rely on the technique introduced in~\cite{HaiLewSol_1-09,HaiLewSol_2-09} to deal with ``rigid'' Coulomb quantum systems, based on the Graf-Schenker inequality.

\begin{remark}
When $\Omega_N=N^{1/3}\bDelta$ where $\bDelta$ is a tetrahedron, we will show that 
\begin{equation}
 \lim_{N\to\ii}\frac{E_\hbar(\rho_0\1_{\Omega_N}\ast \chi)}{|\Omega_N|}
=\rho_0^{4/3}\,e_{\rm UEG}\big(\hbar^2\rho_0^{1/3}\big) 
 \label{eq:thermodynamic_limit_quantum_simplices}
\end{equation}
for every fixed $\chi\geq0$ of compact support with $\int\chi=1$ and $\int_{\R^3}|\nabla\sqrt\chi|^2<\ii$. Since 
$$\rho_0\1_{\Omega_N^-}\leq \rho_0\1_{\Omega_N}\ast \chi\leq \rho_0\1_{\Omega_N^+}$$ 
as soon as $\ell_N$ is much larger than the size of the support of $\chi$, this is an upper bound to~\eqref{eq:thermodynamic_limit_quantum}. We expect that~\eqref{eq:thermodynamic_limit_quantum_simplices} holds for a general scaled set $\Omega_N$, but our proof does not provide this limit.
\end{remark}

\begin{proof}
By scaling we can assume $\rho_0=1$ throughout the proof. For simplicity of notation, we also assume that $q=1$. 

\medskip

\noindent\textbf{Step 1. Preliminary bounds.}
We start by showing that for any smooth-enough $\rho$ which is equal to one in a neighborhood of a regular domain $\Omega$, the energy $E_\hbar(\rho)$ is bounded above by a constant times $|\Omega|$. For this we have to construct a trial state having this density $\rho$ and a kinetic energy of the order of $\Omega$. This might be involved in the canonical case, but is easy in the grand-canonical case where we can resort to quasi-free states.

\begin{lemma}[A priori bounds]
Let $0\leq \rho\leq 1$ be an arbitrary function such that $\sqrt{\rho}\in H^1(\R^3)$. Then 
\begin{equation}
\hbar^2\int_{\R^3}|\nabla\sqrt{\rho}|^2+E_0(\rho)\leq E_\hbar(\rho)\leq \hbar^2\int_{\R^3}|\nabla\sqrt{\rho}|^2+\hbar^{2}c_{\rm TF}\int_{\R^3}\rho.
\label{eq:apriori_bounds}
\end{equation}
\end{lemma}

\begin{proof}
For $\sqrt{\rho}\in H^1(\R^3)$ with $0\leq \rho\leq1$, we can use as trial state the unique quasi-free state $(\Gamma_n)_{n\geq0}$ on Fock space that has the one-particle density matrix
$$\gamma=\sqrt{\rho}\; \1(-\Delta\leq 5c_{\rm TF}/3)\,\sqrt{\rho},$$
see~\cite{BacLieSol-94}. Here $\sqrt{\rho}$ is understood in the sense of multiplication operators. Due to the assumption that $0\leq\rho\leq1$, we have $0\leq\gamma\leq1$ in the sense of operators, as is required for fermions. 
In terms of kernels the previous definition can be written as
$$\gamma(x,y)=(2\pi)^{-3/2}\sqrt{\rho(x)}\, f(x-y)\sqrt{\rho(y)},$$
where $\widehat{f}$ is the characteristic function of the ball of radius $\sqrt{5c_{\rm TF}/3}$, with $c_{\rm TF}$ chosen such that $(2\pi)^{-3/2}f(0)=1$, hence $\rho_\gamma(x)=\gamma(x,x)=\rho(x)$.
The indirect energy of this state is
\begin{multline}
\hbar^2\tr(-\Delta)\gamma-\frac12\int_{\R^3}\int_{\R^3}\frac{|\gamma(x,y)|^2}{|x-y|}dx\,dy\\
=\hbar^2\int_{\R^3}|\nabla\sqrt{\rho}|^2+\hbar^{2}c_{\rm TF}\int_{\R^3}\rho-\frac1{2(2\pi)^3}\int_{\R^3}\int_{\R^3}\frac{\rho(x)\rho(y)|f(x-y)|^2}{|x-y|}dx\,dy.
\label{eq:energy_quasi_free}
\end{multline}
Therefore, we have
$$E_\hbar(\rho)\leq \hbar^2\int_{\R^3}|\nabla\sqrt{\rho}|^2+\hbar^{2}c_{\rm TF}\int_{\R^3}\rho.$$
For the lower bound we use the Hoffman-Ostenhof inequality~\cite{Hof-77} for the kinetic energy and the fact that the Coulomb indirect energy can be bounded from below by the classical energy $E_0(\rho)$. This gives 
\begin{equation}
E_\hbar(\rho)\geq \hbar^2\int_{\R^3}|\nabla\sqrt{\rho}|^2+E_0(\rho)
\label{eq:bound_below_classical}
\end{equation}
as claimed. 
\end{proof}

\begin{remark}\label{rmk:upper_bound_large_rho}
Let us take $\rho_N=\1_{\Omega_N}\ast\chi$ for a fixed function $\chi\in C^\ii_c(\R^3,\R^+)$ with $\int_{\R^3}\chi=1$ and $\int_{\R^3}|\nabla\sqrt\chi|^2<\ii$. Then $\rho_N$ is equal to $1$ on the inner approximation $\Omega_N^-$ and 0 outside of $\Omega_N^+$, which are defined as in the statement of Theorem~\ref{thm:UEG_thermo_limit_quantum}. In this case we even have
$$\frac1{2(2\pi)^3}\int_{\R^3}\int_{\R^3}\frac{\rho_N(x)\rho_N(y)|f(x-y)|^2}{|x-y|}dx\,dy=|\Omega_N|c_{\rm D}+o(|\Omega_N|).$$
In addition, $\nabla\rho_N=\1_{\Omega_N}\ast \nabla\chi$ is bounded by
$$|\nabla\rho_N|=|\1_{\Omega_N}\ast \nabla\chi|\leq 2\sqrt{\rho_N}\sqrt{\1_{\Omega_N}\ast |\nabla\sqrt{\chi}|^2}\leq 2\sqrt{\rho_N}\left(\int_{\R^3}|\nabla\sqrt\chi|^2\right)^{1/2}$$
and has its support in $\Omega_N^+\setminus\Omega_N^-$, a set which has a volume negligible compared with $|\Omega_N|$ due to the regularity of the set $\Omega_N$. So we get
$$\int_{\R^3}|\nabla\sqrt{\rho_N}|^2=\int_{\R^3}\frac{|\nabla\rho_N|^2}{4\rho_N}=o(|\Omega_N|)$$
and therefore find that
\begin{equation}
E_\hbar(\rho_N)\leq \left(\hbar^{2}c_{\rm TF}-c_{\rm D}\right)|\Omega_N|+o(|\Omega_N|).
\label{eq:bound_above_TFD}
\end{equation}
After passing to the limit $N\to\ii$, we get 
\begin{equation}
\limsup_{N\to\ii}\frac{E_\hbar(\1_{\Omega_N}\ast\chi)}{|\Omega_N|}\leq \hbar^{2}c_{\rm TF}-c_{\rm D}
\end{equation}
which is the upper bound in~\eqref{eq:limit_quantum_classical_large_rho}.
\end{remark}

\medskip

\noindent\textbf{Step 2. Limit for simplices.} We now use the smeared version of the Graf-Schenker inequality in order to prove the convergence of the energy per unit volume, in the special case of simplices. 

\begin{lemma}[Smeared Graf-Schenker inequality for the exchange-correlation energy]\label{lem:smeared_Graf_Schenker}
Let $\bDelta$ be a tetrahedron and $\eta\in C^\ii_c(\R^3,\R^+)$ be a fixed function such that $\int_{\R^3}\eta=1$ and $\int_{\R^3}|\nabla\sqrt\eta|^2<\ii$. Then there exists a constant $\kappa$ such that for every $\ell\geq2\kappa$ and every $N$-particle symmetric probability $\bP$, we have
\begin{multline}
C(\bP)-D(\rho_{\bP},\rho_{\bP})\\
\geq \frac{1-\kappa/\ell}{|\ell\bDelta|}\int_{\R^3\times SO(3)}\left\{C\big(\bP_{|\chi_{g\ell\bDelta}}\big)-D(\rho_{\bP}\chi_{g\ell\bDelta}^2,\rho_\bP\chi_{g\ell\bDelta}^2)\right\}\,dg\\-\kappa\frac{N}{\ell}(1+\|\rho_{\bP}\|_{L^\ii(\R^3)})
\label{eq:lower_bd_Graf_Schenker_smooth}
\end{multline}
where $\bP_{|\chi_{g\ell\bDelta}}$ is the (grand-canonical) restriction of $\bP$ associated with the localization function $\chi_{g\ell\bDelta}:=\sqrt{\1_{g\ell\bDelta}\ast\eta}$.
\end{lemma}

\begin{proof}
Lemma 6 in~\cite{GraSch-95} tells us that for any radial function $\eta\in C^\ii_c(\R^3,\R^+)$ with $\int_{\R^3}\eta=1$, there exists a constant $\kappa$ such that the potential 
$$\tilde w_\ell(x)=\frac1{|x|}-\left(1-\frac{\kappa}\ell\right)\frac{\tilde h_\ell(x)}{|x|}-\frac{\kappa}{\ell}\frac{1}{|x|\left(1+|x|\right)}$$
has positive Fourier transform for all $\ell>2\kappa$, where 
\begin{align*}
\tilde h_\ell(x-y)&=\frac1{|\ell\bDelta|}\int_{SO(3)}\big(\1_{\ell\bDelta}\ast\eta\big)\ast\big(\1_{-\ell\bDelta}\ast\eta\big)(Rx-Ry)\,dR\\
&=\frac1{|\ell\bDelta|}\int_{SO(3)}\int_{\R^3}\big(\1_{R^{-1}\ell \bDelta+z}\ast\eta\big)(y)\big(\1_{R^{-1}\ell \bDelta+z}\ast\eta\big)(x)\,dz\,dR
\end{align*}
and with $\bDelta$ a tetrahedron. Similarly as in Lemma~\ref{lem:Graf-Schenker}, we find that
\begin{align}
&C(\bP)-D(\rho_{\bP},\rho_{\bP})\nn\\
&\qquad\geq \frac{1-\kappa/\ell}{|\ell\bDelta|}\int_{\R^3\times SO(3)}\left\{C\big(\bP_{|\chi_{g\ell\bDelta}}\big)-D(\rho_{\bP}\chi_{g\ell\bDelta}^2,\rho_\bP\chi_{g\ell\bDelta}^2)\right\}\,dg-c\frac{N}{\ell}\nn\\
&\qquad\qquad +\frac{\kappa}\ell \left(\pscal{\sum_{1\leq j<k\leq N}W(x_j-x_k)}_\bP-D_W(\rho_\bP,\rho_\bP)\right)
\label{eq:lower_bd_Graf_Schenker_quantum}
\end{align}
where $\chi_{g\ell\bDelta}=\sqrt{\1_{g\ell\bDelta}\ast\eta}$ and $W(x)=|x|^{-1}(1+|x|)^{-1}$. In order to estimate the second term from below, we could use a part of the kinetic energy as in~\cite{ConLieYau-88,GraSch-95}. Here the situation is easier since we can use the additional information that $\rho$ is bounded. Our strategy is to replace $W$ by the short range potential $Y(x)=e^{-\sqrt{2}|x|}/|x|$ in a lower bound and then use that 
$$D_Y(\rho_\bP,\rho_\bP)\leq \frac12\|\rho\|_{L^\ii(\R^3)}N\int_{\R^3}Y.$$
For the lower bound we remark that $W-Y$ is positive and has positive Fourier transform. Indeed, writing 
$$\frac{1}{|x|(1+|x|)}=\int_0^\ii \frac{e^{-t(1+|x|)}}{|x|}\,dt=\int_0^\ii \frac{e^{-t|x|}}{|x|}e^{-t}\,dt,$$
we see that 
$$\widehat{W}(k)-\widehat{Y}(k)=\sqrt{\frac2\pi}\left(\int_0^\ii \frac{e^{-t}}{|k|^2+t^2}\,dt-\frac1{|k|^2+2}\right),$$
which is positive by Jensen's inequality. 
So, using that $\lim_{r\to0}(W(r)-Y(r))=\sqrt{2}$ and arguing again as in the proof of Lemma~\ref{lem:Graf-Schenker} we conclude that 
\begin{align*}
&\pscal{\sum_{1\leq j<k\leq N}W(x_j-x_k)}_\bP-D_W(\rho_\bP,\rho_\bP)\\
&\qquad\qquad \geq \pscal{\sum_{1\leq j<k\leq N}Y(x_j-x_k)}_\bP-D_Y(\rho_\bP,\rho_\bP)-\frac{N}{\sqrt2}\\
&\qquad\qquad \geq -\left(\frac{1}{\sqrt2}+\frac12\|\rho\|_{L^\ii(\R^3)}\int_{\R^3}Y\right)N.
\end{align*}
\end{proof}

Now we are able to prove the existence of the thermodynamic limit for simplices. 

\begin{lemma}[Thermodynamic limit for simplices]\label{lem:limit_simplices}
Let $\bDelta\subset\R^3$ be any simplex containing $0$ and $\eta\in C^\ii_c(\R^3,\R^+)$ be a radial function with $\int_{\R^3}\eta=1$ and $\int_{\R^3}|\nabla\sqrt\eta|^2<\ii$. Then the limits
\begin{multline}
 \lim_{\substack{L_N\to\ii\\ \frac{\ell_N}{L_N}\to0}}\left(\inf_{\substack{\sqrt{\rho}\in H^1(\R^3)\\ \1_{(L_N-\ell_N)\bDelta}\leq \rho\leq \1_{(L_N+\ell_N)\bDelta} }}\frac{E_\hbar(\rho)}{(L_N)^3|\bDelta|}\right)\\
 =\lim_{L_N\to\ii}\frac{E_\hbar(\1_{L_N\bDelta}\ast\eta)}{(L_N)^3|\bDelta|}
=\,e_{\rm UEG}\big(\hbar^2\big) 
 \label{eq:thermodynamic_limit_quantum_tetrahedron}
\end{multline}
exist and are independent of the simplex $\bDelta$, of $\eta$ and of the sequences $L_N,\ell_N\to\ii$. 
\end{lemma}

\begin{proof}
We use the same notation as in Lemma~\ref{lem:smeared_Graf_Schenker} and its proof. 
From the IMS formula, we have on $\R^{3N}$
$$-\sum_{j=1}^N\Delta_{x_j}= \frac{1}{|g\ell\bDelta|}\int_{\R^3\times SO(3)}\chi_{g\ell\bDelta}\left(-\sum_{j=1}^N\Delta_{x_j}\right)\chi_{g\ell\bDelta}-\underbrace{N\frac{\int_{\R^3}|\nabla\chi_{\ell\bDelta}|^2}{|\ell\bDelta|}}_{=O(N/\ell^2)}$$
(see~\cite[Lemma~7]{GraSch-95} and~\cite[Eq.~(30)]{HaiLewSol_2-09}). We then need the notion of quantum localized states $\Gamma_{|\chi_{g\ell\bDelta}}$ which is similar to~\eqref{eq:def_localized} (with a partial trace instead of an integral) and is recalled for instance in~\cite{HaiLewSol_1-09,Lewin-11}.
Using~\eqref{eq:lower_bd_Graf_Schenker_smooth} we find for any grand-canonical state $\Gamma$ with $0\leq \rho_\Gamma\leq1$
\begin{equation}
 \cE_\hbar(\Gamma)\geq \frac{1-\kappa/\ell}{|\ell\bDelta|}\int_{\R^3\times SO(3)}\cE_\hbar(\Gamma_{|\chi_{g\ell\bDelta}})-\frac{c}{\ell}\int_{\R^3}\rho_\Gamma. 
 \label{eq:lower_bound_simplices_quantum}
\end{equation}
Here we have introduced for shortness the quantum indirect energy
\begin{multline*}
\cE_\hbar(\Gamma)=\sum_{n=1}^\ii \tr_{L^2_a((\R^3\times\{1,...,q\})^n)}\Bigg(-\hbar^2\sum_{j=1}^n\Delta_{x_j}+\sum_{1\leq j<k\leq n}\frac{1}{|x_j-x_k|}\Bigg)\Gamma_n\\-\frac12\int_{\R^3}\int_{\R^3}\frac{\rho_\Gamma(x)\rho_\Gamma(y)}{|x-y|}dx\,dy
\end{multline*}
of any grand-canonical state $\Gamma$.

Using~\eqref{eq:lower_bound_simplices_quantum}, it is not difficult to see that the two limits in~\eqref{eq:thermodynamic_limit_quantum_tetrahedron} exist and coincide. Indeed, let us introduce 
$$u_N=\inf_{\substack{\sqrt{\rho}\in H^1(\R^3)\\ \1_{(L_N-\ell_N)\bDelta}\leq \rho\leq \1_{(L_N+\ell_N)\bDelta} }}\frac{E_\hbar(\rho)}{(L_N)^3|\bDelta|}\quad\text{and}\quad v(\ell)=\frac{E_\hbar(\1_{\ell\bDelta}\ast\eta)}{\ell^3|\bDelta|}.$$
Since $\eta$ is fixed we have $u_N\leq v(L_N)$ for $N$ large enough. Let now $\Gamma$ be any state satisfying the constraints in the definition of $u_N$, that is, 
$$\1_{(L_N-\ell_N)\bDelta}\leq \rho_{\Gamma}\leq \1_{(L_N+\ell_N)\bDelta}.$$
Let then $\ell\ll L_N$. The set of all the translations and rotations $g\in \R^3\times SO(3)$ such that $g\ell\bDelta+{\rm supp}(\chi)\subset \1_{(L_N-\ell_N)\bDelta}$ has a measure of the order of $|L_N\bDelta|$. More precisely, the set of all the $g$ such that $\rho_\Gamma$ is not one or 0 on the support of $\chi_{g\ell\bDelta}$ has a measure bounded by a constant times $(\ell+\ell_N)(L_N)^2$. For a $g$ such that $\rho_\Gamma\equiv1$ on the support of $\chi_{g\ell\bDelta}$ we can use the rotation and translation invariance of $E_\hbar$ to infer
$$\cE_\hbar(\Gamma_{|\chi_{g\ell\bDelta}})\geq E_\hbar(\1_{g\ell\bDelta}\ast\eta)=E_\hbar(\1_{\ell\bDelta}\ast\eta)= v(\ell)|\ell\bDelta|,$$
since the density of the localized state is by definition $\rho_{\Gamma_{|\chi_{g\ell\bDelta}}}=\rho_\Gamma\chi_{g\ell\bDelta}^2=\1_{g\ell\bDelta}\ast\eta$.
For all the other $g$ for which $\rho_\Gamma\neq0$, we can simply use~\eqref{eq:apriori_bounds} and the Lieb-Oxford inequality, which tells us that 
$$\cE_\hbar(\Gamma_{|\chi_{g\ell\bDelta}})\geq-c|\ell\bDelta|.$$
In total we get the lower bound
$$\frac{\cE_\hbar(\Gamma)}{(L_N)^3|\bDelta|}\geq \left(1-\frac{\kappa}{\ell}\right)\left(1-c\frac{\ell+\ell_N}{L_N}\right)v(\ell)-c\frac{\ell+\ell_N}{L_N}.$$
Minimizing over all $\Gamma$, we get 
\begin{equation}
u_N\geq \left(1-\frac{\kappa}{\ell}\right)\left(1-c\frac{\ell+\ell_N}{L_N}\right)v(\ell)-c\frac{\ell+\ell_N}{L_N}.
\label{eq:relation_u_v}
\end{equation}
By~\eqref{eq:apriori_bounds} we know that $u_N$ and $v(\ell)$ are uniformly bounded. The inequality~\eqref{eq:relation_u_v} tells us that
$$\liminf_{N\to\ii}u_N\geq \limsup_{\ell\to\ii}v(\ell)$$
and since $u_N\leq v(L_N)$ the two sequences have the same limit $e_{\rm UEG}(\hbar^2)$.
\end{proof}

\medskip

\noindent\textbf{Step 3. Limit for an arbitrary sequence of domains.} Next we prove that for any domain $\Omega_N$ satisfying the assumptions of the statement, the limit exists and is the same as for simplices. 

The lower bound is proved in exactly the same way as for simplices. Using~\eqref{eq:lower_bound_simplices_quantum} and the assumptions on the regularity of $\Omega_N$, we find a lower bound similar to~\eqref{eq:relation_u_v},
\begin{equation}
\inf_{\substack{\sqrt{\rho}\in H^1(\R^3)\\ \1_{\Omega_N^-}\leq \rho\leq \1_{\Omega_N^+} }}\frac{E_\hbar(\rho)}{|\Omega_N|}\geq \left(1-\frac{\kappa}{\ell}\right)\left(1-c\frac{\ell+\ell_N}{L_N}\right)v(\ell)-c\frac{\ell+\ell_N}{L_N},
\label{eq:lower_bound_Omega_N}
\end{equation}
where we recall that 
$$v(\ell)=\frac{E_\hbar(\1_{\ell\bDelta}\ast\eta)}{\ell^3|\bDelta|}$$
is the quantum energy of a simplex, smeared-out with the fixed function $\eta$. Passing to the limit $N\to\ii$ and then $\ell\to\ii$ using Theorem~\ref{lem:limit_simplices} gives the lower bound
\begin{equation}
\liminf_{N\to\ii}\inf_{\substack{\sqrt{\rho}\in H^1(\R^3)\\ \1_{\Omega_N^-}\leq \rho\leq \1_{\Omega_N^+} }}\frac{E_\hbar(\rho)}{|\Omega_N|}\geq  \,e_{\rm UEG}\big(\hbar^2\big) .
\end{equation}

The upper bound is more complicated. The method introduced in~\cite{HaiLewSol_1-09} works here but, unfortunately, we cannot directly apply the abstract theorem proved in~\cite{HaiLewSol_1-09}, because the assumption (A4) of~\cite{HaiLewSol_1-09} is not obviously verified in our situation (the assumption (A4) essentially requires that the energy be subadditive up to small errors). So, instead we follow the proof of~\cite[p.~475--483]{HaiLewSol_1-09} line-by-line and bypass (A4) at the only place where it was used in~\cite{HaiLewSol_1-09}. For shortness we only explain the difference without providing all the details of the proof. 

Similarly to what we have done in the proof of Theorem~\ref{thm:UEG_thermo_limit}, the idea is to use one big simplex $\bDelta'_N\supset \Omega_N$ of volume proportional to $\Omega_N$ together with a tiling of small simplices of side length $\ell\ll \ell_N\ll L_N$. The first step is to replace $\Omega_N$ by its inner approximation $A_N$ which is the union of all the simplices that are inside $\Omega_N$. More precisely, this amounts to replacing the optimal $\rho_N$ satisfying the constraint $\1_{\Omega_N^-}\leq \rho_N\leq \1_{\Omega_N^+}$ by $(\1_{A_N}\ast\eta)\rho_N$.
In~\cite[Eq.~(42)]{HaiLewSol_1-09} the property (A4) was used to deduce that
$$E_\hbar(\Omega_N)\leq E_\hbar(A_N)+o(|\Omega_N|).$$
Here the new density $\rho_N(\1_{A_N}\ast\eta)$ satisfies the constraint 
$$\1_{\Omega_N^-}\leq (\1_{A_N}\ast\eta)\rho_N\leq \1_{\Omega_N^+}$$ 
because $\1_{A_N}\ast\eta$ takes values in $(0,1)$ only at a distance from the boundary of $\Omega_N$ proportional to $\ell\ll\ell_N$. The definition of $\Omega_N$ with the minimum over all the densities satisfying $\1_{\Omega_N^-}\leq \rho\leq \1_{\Omega_N^+}$ implies immediately that the energy goes up:
$$E_\hbar(\rho_N)\leq E_\hbar(\1_{A_N}\ast\eta).$$
The rest of the proof then follows that of~\cite{HaiLewSol_1-09} \emph{mutatis mutandis}.

\medskip

Since the quantum energy $\cE_\hbar$ is linear and increasing in $\hbar^2$, the minimum $E_\hbar(\rho)$ is concave non-decreasing in $\hbar^2$. By passing to the pointwise limit we obtain that the limit $e_{\rm UEG}(\lambda)$ is concave non-decreasing in $\lambda=\hbar^2$. It remains to show the limits~\eqref{eq:limit_quantum_classical_small_rho} and~\eqref{eq:limit_quantum_classical_large_rho} of $e_{\rm UEG}(\lambda)$ at small and large $\lambda$.

\medskip

\noindent\textbf{Step 4. Limit \eqref{eq:limit_quantum_classical_large_rho} as $\lambda\to\ii$.}  From Remark~\ref{rmk:upper_bound_large_rho}, we have 
$$e_{\rm UEG}\big(\lambda)\leq c_{\rm TF}\lambda-c_{\rm D}.$$
In order to prove the lower bound, we consider a large simplex $L\bDelta$ and $\Gamma$ a quantum state minimizing $E_\hbar(\1_{L\bDelta}\ast\eta)$. Then we can write
\begin{multline}
\cE_\hbar(\Gamma)=\sum_{n=1}^\ii \tr_{L^2_a((\R^3\times\{1,...,q\})^n)}\Bigg(-\hbar^2\sum_{j=1}^n\Delta_{x_j}+\sum_{1\leq j<k\leq n}\frac{1}{|x_j-x_k|}\Bigg)\Gamma_n\\-\int_{L\bDelta}\int_{\R^3}\frac{\rho_\Gamma(y)}{|x-y|}\,dx\,dy+\frac12\int_{L\bDelta}\int_{L\bDelta}\frac{dx\,dy}{|x-y|}\\
-\frac12\int_{\R^3}\int_{\R^3}\frac{\big(\rho-\1_{L\bDelta}\big)(x)\big(\rho-\1_{L\bDelta}\big)(y)}{|x-y|}dx\,dy.
\label{eq:decomp_Jellium}
\end{multline}
The last term is proportional to
\begin{multline*}
L^5\int_{\R^3}\frac{|\widehat{\1_{\bDelta}}(k)|^2|1-(2\pi)^{3/2}\widehat{\eta}(k/L)\big|^2}{|k|^2}\,dk\\
=L^3\int_{\R^3}|\widehat{\1_{\bDelta}}(k)|^2\bigg|\frac{k}{|k|}\cdot\underbrace{\int_{\R^3}x\,\eta(x)\,dx}_{=0}\bigg|^2\,dk+o(L^3)
\end{multline*}
and hence disappears in the thermodynamic limit.
Now if we forget the constraint that $\rho_\Gamma=\1_{L\bDelta}\ast\eta$ and take the thermodynamic limit, we get the Jellium energy $e_{\rm Jel}(\lambda)$ which was studied in~\cite{LieNar-75,GraSol-94}:
\begin{multline*}
e_{\rm Jel}(\lambda)\\=\lim_{L\to\ii}\frac{1}{|L\bDelta|}\inf_{\Gamma}\bigg\{\sum_{n=1}^\ii \tr_{L^2_a((\R^3\times\{1,...,q\})^n)}\Bigg(-\hbar^2\sum_{j=1}^n\Delta_{x_j}+\sum_{1\leq j<k\leq n}\frac{1}{|x_j-x_k|}\Bigg)\Gamma_n\\
-\int_{L\bDelta}\int_{\R^3}\frac{\rho_\Gamma(y)}{|x-y|}\,dx\,dy+\frac12\int_{L\bDelta}\int_{L\bDelta}\frac{dx\,dy}{|x-y|}\bigg\}.
\end{multline*}
Coming back to~\eqref{eq:decomp_Jellium} we get
$$e_{\rm UEG}(\lambda)\geq e_{\rm Jel}(\lambda).$$ 
Graf and Solovej have proved in~\cite{GraSol-94} that
$$e_{\rm Jel}(\lambda)\geq c_{\rm TF}\lambda-c_{\rm D}+O(\lambda^{-1/5+\eps}).$$
Strictly speaking,~\cite{GraSol-94} deals with the canonical case but the proof works exactly the same for the grand-canonical energy.
Hence we immediately obtain~\eqref{eq:limit_quantum_classical_large_rho}.

\medskip

\noindent\textbf{Step 5. Limit \eqref{eq:limit_quantum_classical_small_rho} as $\lambda\to0$.}
Next we turn to the proof of~\eqref{eq:limit_quantum_classical_small_rho}, for which we only have to derive the upper bound, since obviously $E_\hbar\geq E_0$. Let $C$ be the unit cube and $C_n=nC$ be the cube of volume $n^3$. Let $\bP$ be an $N$-particle probability such that $\rho_\bP=\1_{C_n}$ and 
$$C(\bP)-\frac12\int_{C_n}\int_{C_n}\frac{dx\,dy}{|x-y|}=E(\1_{C_n}).$$
It is proved in~\cite{ButChaPas-17} that there exists a $\delta'=\delta'(n)>0$ such that $\bP$ is supported on the set where all the particles stay at a distance $\delta'$ from each other, that is,  $|x_j-x_k|\geq\delta'$ for $j\neq k$, $\bP$-almost surely. Although we expect that $\delta'$ is independent of $n$, the argument in~\cite{ButChaPas-17} only gives $\delta'(n)>c/n^4$. This is sufficient for our proof since $n$ will be fixed until the end of the argument.

The idea of the proof is to place $k^3$ copies of this cube in order to build a much larger cube  of volume $k^3n^3$, and then to construct a quantum state by replacing each pointwise particle located at $x_j$ by a quantum one having density $\chi^2(\cdot-x_j)$, where $\chi$ is a smooth function with compact support. Unfortunately, some of the particles can get close to the particles of another cube when they approach the boundary, and the overlap of the functions create some normalization issues. The particles can form clusters of at most 8 particles, when they are in a corner of a cube. It is possible to orthogonalize the overlapping quantum states by using the recent smooth extension of the Hobby-Rice theorem proved in~\cite{LazLie-13,Rutherfoord-13}, which was motivated by the representability of currents in density functional theory~\cite{LieSch-13}. But so far there does not exist any estimate on the resulting kinetic energy.

In order to bypass this difficulty, we insert a layer of unit cubes between the different copies of $C_n$, as displayed in Figure~\ref{fig:trial_state_UEG}, and form a slightly larger cube $C'_k$ of volume $(nk+k+1)^3$. We call $p\simeq n^2k^3$ the number of such unit cubes and $r_1,...,r_p$ their centers. In these cubes we place the particles on a subset of the cubic lattice and average over the positions of this lattice as was done in~\cite{LewLie-15}. In other words, we use the strongly correlated $p$-particle probability density
$$d\bQ(y_1,...,y_p)=\int_{[-1/2,1/2]^3}\delta_{r_1+\tau}(y_1)\cdots \delta_{r_p+\tau}(y_p)\,d\tau$$
which has the constant density $\rho_{\bQ}=\sum_{j=1}^p\1_{r_j+[-1/2,1/2]^3}$. Finally, we denote by $\bP'_k$ the tensor product of $\bQ$ and of the $k^3$ independent copies of $\bP$. With this construction we have gained that the clusters can never contain more than two particles at a distance $\leq\delta'$ from each other, instead of 8. This allows us to use the simpler orthonormalization procedure of~\cite{Harriman-81,Lieb-83b}. Since the volume occupied by the corridors is small compared to the overall volume, this will only generate a small error in our estimate.

\begin{figure}[h]
\includegraphics[width=10cm]{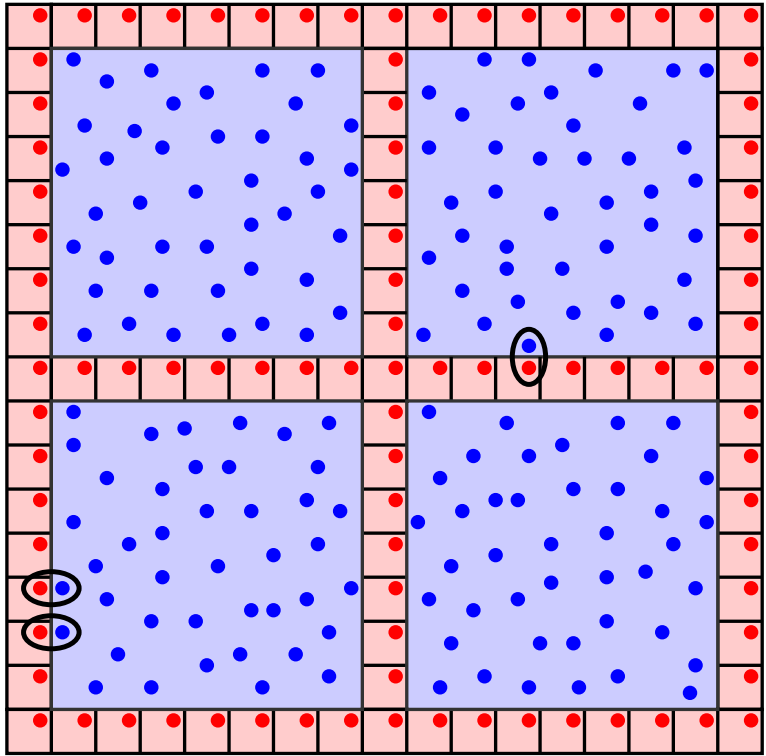}
\caption{The trial state used to show the convergence $e_{\rm UEG}(\lambda)\to e_{\rm UEG}$ as $\lambda\to0$ (semi-classical limit). \label{fig:trial_state_UEG}}
\end{figure}

Let us choose for instance $\chi=\sqrt{\frac{105}{32\pi}}(1-|x|^2)_+$, with the constant chosen such that $\int_{\R^3}\chi^2=1$. Let $\chi_\delta(x)=\delta^{-3/2}\chi(x/\delta)$ where $\delta$ is any fixed number such that $0<\delta\leq \min(\delta',1)/3$. For all positions $x_1,...,x_N$ of the $N$ particles, the functions $\chi_\delta(\cdot-x_i)$ are orthogonal to each other except when the supports of the $\chi_\delta(\cdot-x_j)$ overlap. We make them orthogonal by using the method of Harriman~\cite{Harriman-81} and Lieb~\cite{Lieb-83b}. From the preceding discussion we can assume that we only have two such functions $\chi_\delta(\cdot-x_1)$ and $\chi_\delta(\cdot-x_2)$ with $|x_1-x_2|\leq\delta$. After a rotation and a translation we may assume that $x_1=0$ and $x_2=|x_2|e_1$. We then take
$$f^{(1)}_{x_1,x_2}(y)=\chi_\delta(y-x_1)=\chi_\delta(y)$$
and
$$f^{(2)}_{x_1,x_2}(y)=\chi_\delta(y-x_2)\exp\left(2i\pi\frac{\dps\int_{(-\ii;y^{(1)}]\times\R^2}\chi_\delta(z)\chi_\delta(z-x_2)\,dz}{\dps\int_{\R^3}\chi_\delta(z)\chi_\delta(z-x_2)\,dz}\right).$$
It can be checked that $\int f^{(1)}_{x_1,x_2}f^{(2)}_{x_1,x_2}=0$.
In order to estimate the kinetic energy, we introduce the function
\begin{align*}
\eta(t)&=\int_{\R^2}\chi_\delta(te_1+u)\chi_\delta(te_1+u-x_2)\,du \\
&=\frac{105}{16\delta}\int_0^\ii\left(1-\left(\frac{t}{\delta}\right)^2-r^2\right)_+\left(1-\left(\frac{t-|x_2|}{\delta}\right)^2-r^2\right)_+r\,dr.
\end{align*}
A calculation gives
$$\int_\R\eta(t)\,dt\underset{|x_2|\to 2\delta^-}{\sim}\frac{315}{512}\left(2-\frac{|x_2|}{\delta}\right)^4.$$
Next we observe that 
$$\int_{\R^3}|\nabla f_{x_1,x_2}^{(2)}|^2\leq 2\int_{\R^3}|\nabla\chi_\delta|^2+8\pi^2\frac{\dps\int_{\R^3} \eta(z^{(1)})^2\chi_\delta(z-x_2)^2\,dz}{\dps\left(\int_{\R}\eta(t)\,dt\right)^2}.$$
and
\begin{multline*}
\int_{\R^3} \eta(z^{(1)})^2\chi_\delta(z-x_2)^2\,dz= \frac{c}\delta \int_\R\eta(t)^2\left(1-\left(\frac{t-|x_2|}{\delta}\right)^2\right)^3_+\,dt\\
\leq\frac{c\left(2-\frac{|x_2|}{\delta}\right)_+^{10}}{\delta^2}. 
\end{multline*}
From this we conclude that $\int_{\R^3}|\nabla f_{x_1,x_2}^{(2)}|^2$ is uniformly bounded with respect to $x_1$ and $x_2$, by a constant times $\delta^{-2}$. 

For all positions $X=(x_1,...,x_N)\in\R^{3N}$ of the $N$ particles, we use the previous construction for all the pairs of particles which are at a distance $\leq\delta$ and denote by $f_X^{(1)},...,f_X^{(N)}$ the corresponding functions. Those are now orthogonal and we can define the Slater determinant
$$\Psi_{X}(y_1,...,y_N)=\frac{1}{\sqrt{N!}}\det\big(f^{(i)}_{X}(y_j)\big).$$
This function satisfies
$$\rho_{|\Psi_{X}|^2}=\sum_{j=1}^N\chi_\delta(y-x_j)^2,$$
and 
\begin{align*}
&\pscal{\Psi_{X},\left(-\hbar^2\sum_{j=1}^N\Delta_{y_j}+\sum_{1\leq j<k\leq N}\frac{1}{|y_j-y_k|}\right)\Psi_{X}}\\
&\qquad=\hbar^2\sum_{j=1}^N\int_{\R^3}|\nabla f_{X}^{(j)}|^2 +\frac12 D\left(\sum_{j=1}^N\chi_\delta(y-x_j)^2,\sum_{j=1}^N\chi_\delta(y-x_j)^2\right)\\
&\qquad\qquad\qquad-\frac{1}{2}\int_{\R^3}\int_{\R^3}\frac{\left|\sum_{j=1}^Nf_{X}^{(j)}(y)\overline{f_{X}^{(j)}(z)}\right|^2}{|y-z|}
\,dy\,dz\\
&\qquad=\hbar^2\sum_{j=1}^N\int_{\R^3}|\nabla f_{X}^{(j)}|^2 +\sum_{1\leq j< k\leq N}D\left(\chi_\delta(\cdot-x_j)^2,\chi_\delta(\cdot-x_k)^2\right)\\
&\qquad\qquad\qquad-\sum_{1\leq j< k\leq N}D\left(f^{(j)}_X\overline{f^{(k)}_X},f^{(j)}_X\overline{f^{(k)}_X}\right).
\end{align*}
Since $D$ is a positive quadratic form, we have
$$D\left(f^{(j)}_X\overline{f^{(k)}_X},f^{(j)}_X\overline{f^{(k)}_X}\right)\geq0$$
for every $j<k$. Since $\chi$ is radial, we have 
$$\int_{\R^3}\frac{\chi_\delta(y)^2}{|y-z|}\,dy=\int_{\R^3}\frac{\chi_\delta(y)^2}{\max(|y|,|z|)}\,dy\leq \frac{\int_{\R^3}\chi_\delta^2}{|z|},$$
by Newton's theorem, and hence
\begin{multline*}
D\left(\chi_\delta(\cdot-x_j)^2,\chi_\delta(\cdot-x_k)^2\right)\\=\int_{\R^3}\int_{\R^3}\frac{\chi_\delta(y-x_j)^2\chi_\delta(z-x_k)^2}{|y-z|}\,dy\,dz \leq \frac{1}{|x_j-x_k|}. 
\end{multline*}
Using that the kinetic energy is bounded uniformly with respect to $x_1,...,x_N$ and that we have of the order of $n^2k^3=N/n$ little cubes, we obtain
\begin{multline*}
\pscal{\Psi_{X},\left(-\hbar^2\sum_{j=1}^N\Delta_{y_j}+\sum_{1\leq j<k\leq N}\frac{1}{|y_j-y_k|}\right)\Psi_{X}}\\
\leq cN\hbar^2\left(1+\frac{1}{\delta^2n}\right)+\sum_{1\leq j<k\leq N}\frac{1}{|x_j-x_k|}.
\end{multline*}
We finally introduce the corresponding quantum state
$$\Gamma=\int_{\R^{3N}}|\Psi_{x_1,...,x_N}\rangle\langle\Psi_{x_1,...,x_N}|\,d\bP'_k(x_1,...,x_N)$$
which has the density $\rho_\Gamma=\rho_{\bP'_k}\ast\chi_\delta^2=\1_{C'_k}\ast\chi_\delta^2$ and the indirect energy
\begin{align*}
\cE_\hbar(\Gamma)&\leq cN\hbar^2\left(1+\frac{1}{\delta^2n}\right)+C(\bP'_k)-D\big(\1_{C'_k}\ast\chi_\delta^2,\1_{C'_k}\ast\chi_\delta^2\big) \\
&= cN\hbar^2\left(1+\frac{1}{\delta^2n}\right)+k^3\;E(\1_{C_n})+E(\bQ)\\
&\qquad\qquad+D(\1_{C'_k},\1_{C'_k})-D\big(\1_{C'_k}\ast\chi_\delta^2,\1_{C'_k}\ast\chi_\delta^2\big).
\end{align*}
The last term is proportional to
$$N^{5/3}\int_{\R^3}\frac{|\widehat{\1_C}(p)|^2}{|p|^2}\left(1-(2\pi)^3|\widehat{\chi_\delta^2}(p/N^{1/3})|^2\right)\,dp=\frac{N}{3}\int_{\R^3}|x|^2\chi_\delta(x)^2\,dx+o(N).$$
Also, we have for a universal constant $c$
$$E(\bQ)=\sum_{1\leq j<k\leq p}\left(\frac{1}{|x|}-\frac{1}{|x|}\ast\1_{C}\ast\1_{C}\right)(r_j-r_k)-\frac{p}{2}D(\1_{C},\1_{C})\leq c n^2k^3$$
since
$$\frac{1}{|x|}-\frac{1}{|x|}\ast\1_{C}\ast\1_{C}=O_{|x|\to\ii}\left(\frac{1}{|x|^4}\right).$$
Taking the limit $k\to\ii$ first we find
\begin{equation}
e_{\rm UEG}(\hbar^2)\leq c\hbar^2\left(1+\frac{1}{\delta^2n}\right)+\frac{c}{n}+\frac{E(\1_{C_n})}{n^3}+\frac{\delta^2}3 \int_{\R^3}|x|^2\chi(x)^2\,dx.
\label{eq:final_estimate_UEG_semi-classical}
\end{equation}
Taking now the limit $\hbar\to0$ we obtain
$$\limsup_{\hbar\to0}e_{\rm UEG}(\hbar^2)\leq \frac{c}{n}+\frac{E(\1_{C_n})}{n^3}+\frac{\delta^2}3 \int_{\R^3}|x|^2\chi_\delta(x)^2\,dx.$$
Here we have to take $\delta\to0$ first (recall that $\kappa$ does not depend on $\delta$, but $\delta\leq\min(\delta'(n),1)/3$), and then $n\to\ii$. We find
$$\limsup_{\hbar\to0}e_{\rm UEG}(\hbar^2)\leq e_{\rm UEG}$$
as we wanted.
\end{proof}


\end{document}